\documentclass[reprint,
superscriptaddress,showkeys,
nofootinbib,
 amsmath,amssymb,
 aps,
prx,
]{revtex4-2}
\usepackage{booktabs}
\usepackage{cmap}% text search
\usepackage{graphicx}% Include figure files
\usepackage{dcolumn}% Align table columns on decimal point
\usepackage{mathtools,physics,bbm,euscript,mathrsfs,enumitem,nicefrac,amsthm,amsfonts,float,tcolorbox,bm}
\usepackage{tikz}
\usetikzlibrary{arrows.meta, positioning, calc, intersections,decorations.markings,plotmarks}
\usepackage{appendix}
\usepackage{placeins}
\makeatletter
\let\newfloat\newfloat@ltx
\makeatother
\usepackage{algorithm}
\usepackage{algpseudocode}
\usepackage{algpseudocodex}
\tikzset{algpxIndentLine/.style={draw=black, thick}} 

\usepackage{subcaption}

\usepackage[%backref,
colorlinks,bookmarks=true, citecolor=blue!90!black, urlcolor=blue!90!black,linkcolor=blue!90!black, pdfencoding=auto, psdextra]{hyperref}
\hypersetup{pdfauthor={Name}}

\usepackage{multirow} 
\newcolumntype{C}[1]{>{\centering\arraybackslash}p{#1}}

\usepackage[capitalize]{cleveref}
\makeatletter
\newcommand{\LongState}[1]{\State \parbox[t]{\linewidth-\ALG@thistlm}{\raggedright #1\strut}}
\makeatother
\usepackage{calc}

\newcommand{\MultilineInput}[1]{%
  \Statex \textbf{Input:}\enskip \parbox[t]{\linewidth-\widthof{\textbf{Input:}\enskip}}{\raggedright #1}
}
\newcommand{\MultilineOutput}[1]{%
  \Statex \textbf{Output:}\enskip \parbox[t]{\linewidth-\widthof{\textbf{Output:}\enskip}}{\raggedright #1}
}
\newtheoremstyle{upright}
  {3pt}   % Space above
  {3pt}   % Space below
  {\normalfont} % Body font
  {}      % Indent amount
  {\bfseries} % Theorem head font
  {.}     % Punctuation after theorem head
  {.5em}  % Space after theorem head
  {}      % Theorem head spec

\theoremstyle{plain}
\newtheorem{theorem}{Theorem}[section]
\newtheorem{lemma}[theorem]{Lemma}

\newtheorem{proposition}[theorem]{Proposition}

\theoremstyle{upright}
\newtheorem{definition}{Definition}[section]
\newtheorem{remark}{Remark}
\newtheorem{example}{Example}
\newtheorem{construction}{Construction}

\newcommand\nc\newcommand
\nc\bfa{{\boldsymbol a}}\nc\bfA{{\boldsymbol A}}\nc\cA{{\EuScript A}}
\nc\bfb{{\boldsymbol b}}\nc\bfB{{\boldsymbol B}}\nc\cB{{\EuScript B}}
\nc\bfc{{\boldsymbol c}}\nc\bfC{{\boldsymbol C}}\nc\cC{{\mathscr C}}\nc\uc{{\underline c}}
\nc\bfd{{\boldsymbol d}}\nc\bfD{{\boldsymbol D}}\nc\cD{{\mathscr D}}
\nc\bfe{{\boldsymbol e}}\nc\bfE{{\boldsymbol E}}\nc\cE{{\EuScript E}}\nc\ue{{\underline e}}
\nc\bff{{\boldsymbol f}}\nc\bfF{{\boldsymbol F}}\nc\cF{{\mathcal F}}\nc\uf{{\underline f}}
\nc\bfg{{\boldsymbol g}}\nc\bfG{{\boldsymbol G}}\nc\cG{{\EuScript G}}
\nc\bfh{{\boldsymbol h}}\nc\bfH{{\boldsymbol H}}\nc\cH{{\mathcal H}}
\nc\bfi{{\boldsymbol i}}\nc\bfI{{\boldsymbol I}}\nc\cI{{\mathcal I}}
\nc\bfj{{\boldsymbol j}}\nc\bfJ{{\boldsymbol J}}\nc\cJ{{\EuScript J}}
\nc\bfk{{\boldsymbol k}}\nc\bfK{{\boldsymbol K}}\nc\cK{{\EuScript K}}
\nc\bfl{{\boldsymbol l}}\nc\bfL{{\boldsymbol L}}\nc\cL{{\EuScript L}}
\nc\bfm{{\boldsymbol m}}\nc\bfM{{\boldsymbol M}}\nc\cM{{\EuScript M}}
\nc\bfn{{\boldsymfol n}}\nc\bfN{{\boldsymbol N}}\nc\cN{{\EuScript N}}\nc\un{{\underline n}}
\nc\bfo{{\boldsymbol o}}\nc\bfO{{\boldsymbol O}}\nc\cO{{\EuScript O}}
\nc\bfp{{\boldsymbol p}}\nc\bfP{{\boldsymbol P}}\nc\cP{{\EuScript P}}
\nc\bfq{{\boldsymbol q}}\nc\bfQ{{\boldsymbol Q}}\nc\cQ{{\EuScript Q}}
\nc\bfr{{\boldsymbol r}}\nc\bfR{{\boldsymbol R}}\nc\cR{{\EuScript R}}\nc\ur{{\underline r}}
\nc\bfs{{\boldsymbol s}}\nc\bfS{{\boldsymbol S}}\nc\cS{{\EuScript S}}
\nc\bft{{\boldsymbol t}}\nc\bfT{{\boldsymbol T}}\nc\cT{{\EuScript T}}
\nc\bfu{{\boldsymbol u}}\nc\bfU{{\boldsymbol U}}\nc\cU{{\EuScript U}}
\nc\bfv{{\boldsymbol v}}\nc\bfV{{\boldsymbol V}}\nc\cV{{\mathscr V}}
\nc\bfw{{\boldsymbol w}}\nc\bfW{{\boldsymbol W}}\nc\cW{{\mathscr W}}
\nc\bfx{{\boldsymbol x}}\nc\bfX{{\boldsymbol X}}\nc\cX{{\EuScript X}}\nc\ux{{\underline x}}
\nc\bfy{{\boldsymbol y}}\nc\bfY{{\boldsymbol Y}}\nc\cY{{\mathscr Y}}\nc\uy{{\underline y}}
\nc\bfz{{\boldsymbol z}}\nc\bfZ{{\boldsymbol Z}}\nc\cZ{{\EuScript Z}}
\nc{\1}{{\mathbbm{1}}}

\nc{\remove}[1]{}

\DeclareSymbolFont{bbold}{U}{bbold}{m}{n}
\DeclareSymbolFontAlphabet{\mathbbold}{bbold}

\DeclareMathOperator{\supp}{supp}

\DeclareMathOperator{\wt}{wt}

\nc\reals{{\mathbb R}}
\nc{\ff}{{\mathbb F}}
\nc{\PP}{{\mathbb P}}
\nc{\complex}{{\mathbb C}}

\newcommand{\ah}{{\mathfrak {a}}}
\newcommand{\bh}{{\mathfrak {b}}}
\newcommand{\uh}{{\mathfrak {u}}}
\newcommand{\ph}{{\mathfrak {p}}}
\nc{\gh}{{\mathfrak {g}}}
\nc{\hh}{{\mathfrak {h}}}
\nc{\fh}{{\mathfrak {f}}}

\usepackage{xcolor}
\definecolor{Qcolor}{RGB}{31,120,180}

\usepackage[normalem]{ulem}
\newcommand\redout{\bgroup\markoverwith{\textcolor{red}{\rule[0.5ex]{2pt}{0.8pt}}}\ULon}

\allowdisplaybreaks

\begin{document}

\title{Breaking the bicycle frame: Coset-based quantum LDPC codes}

\author{Arda Aydin}
 \email{aaydin@umd.edu}
\affiliation{ISR and Department of ECE, University of Maryland, College Park, MD 20742
}
\author{Itzhak Tamo}
\email{tamo@tauex.tau.ac.il}
\affiliation{ Department of Electrical Engineering - Systems, Tel Aviv University, Tel Aviv 6997801, Israel}

\author{Alexander Barg}%
 \email{abarg@umd.edu}
\affiliation{ISR and Department of ECE, University of Maryland, College Park, MD 20742
}
\affiliation{Joint Center for Quantum Information and Computer Science,
NIST/University of Maryland, College Park, MD 20742}

\begin{abstract}
Generalizing the  construction of two-block group algebra (2BGA) codes, we introduce a family of two-block quantum LDPC codes constructed using 
the action of a group on the cosets of its subgroup. This replaces the regular group actions of the earlier two-block constructions and significantly expands the search space, yielding new quantum LDPC codes outside the 2BGA family. Through a computer search, we identify several new quantum LDPC codes, including weight-6 codes with parameters $[[48,8,6]]$, $[[96,8,10]]$, and $[[224,12,16]]$, as well as weight-8 codes with parameters $[[84,16,8]]$, $[[112,16,10]]$, $[[128,16,12]]$, and $[[168,16,15]]$. Furthermore, we introduce a maximally packed syndrome extraction schedule of depth $w+2$, including initialization and measurement steps, for any code with a maximum stabilizer weight of $w$ from our family. Under a standard circuit-level noise model, our codes, when decoded using BP-OSD, perform competitively with BB codes, achieving thresholds of $\approx0.65\%$ for the weight-6 family and $\approx0.35\%$ for the weight-8 family. Finally, we introduce a group-theoretic framework to generate sequences of graph-based covers of 2BGA codes, recovering and extending recent results on code constructions of this type.
\end{abstract}
           
\maketitle

\section{\label{sec:Introduction} Introduction}

Quantum low-density parity-check (LDPC) codes are among the leading approaches to realizing fault-tolerant quantum computation.  Surface codes  \cite{KITAEV20032,1997RuMaS..52.1191K,bravyi1998quantumcodeslatticeboundary} form one of the earliest examples of quantum LDPC codes. These codes provide strong error suppression but incur a large qubit overhead, requiring many physical qubits. Specifically, to encode $k$ logical qubits with a distance $d$, the surface code requires $\Theta(k d^2)$ physical qubits. To reduce this overhead, researchers have studied quantum LDPC codes with more complex structures,  aiming at constructing codes that simultaneously achieve a high encoding rate and large distance. The first advances toward achieving this goal were made about 15 years ago \cite{5205648}, followed by further developments in \cite{Hastings21,9490244}, and culminating in the construction of asymptotically good families of quantum LDPC codes \cite{10.1145/3519935.3520017,9996782}.

Asymptotically good codes are essential for the realization of a large-scale, fault-tolerant quantum computation. However, for near-term practical considerations, finite-length quantum codes with specific structural properties play a critical role. For instance, despite its large qubit overhead, the surface code 
is one of the primary choices of quantum LDPC code for practical implementations because of its advantageous geometrical properties. 
Its structure is embeddable in a 2D grid, requiring only local connections. Due to these characteristics, along with its high threshold and robust error suppression under circuit-level noise, surface codes have attracted sustained attention in recent years.

Another promising class of finite-length quantum LDPC codes is the family of bivariate bicycle (BB) codes, introduced in Ref.~\cite{Bravyi2024}. Rather than embedding qubits exclusively in a planar 2D grid, the authors of \cite{Bravyi2024} considered a bi-planar structure. By incorporating a few long-range connections alongside the local connections utilized by the toric code, they successfully constructed quantum LDPC codes with a significantly lower qubit overhead than the surface code, while maintaining comparable thresholds and error suppression capabilities.
The building blocks of BB codes are permutation matrices representing the regular action of the finite abelian group $ \mathbb{Z}_l \times \mathbb{Z}_m$, where $l,m \geq 2$ are integers. 

\begin{figure*}[t]
     \centering
     \begin{tikzpicture}[rotate=90]
         % This work
         \draw[thick, fill=gray!10] (0,0) ellipse (4cm and 6.5cm);
         \node [align=center, font=\large\bfseries] at (0.0, 5.25) {This work};
         \draw[thick, rounded corners=0.75cm] (-4.2, -6.75) rectangle (4.25, 8.0);
         \node[align=center, font=\large] at (3.0, 6.5) {Two-Block \\CSS Codes};
          % 2BGA
         \draw[thick, fill=white!10] (0,-1.25) ellipse (3.75cm and 5.25cm);
        \node[align=center] at (0.0, 3.25) {2BGA \\ Code\cite{Pryadko2BGA}};
         %ZSZ Codes
         \draw[thick, fill=white!10] (1.25, 2.25) ellipse (1.25cm and 0.8cm);
         \node[align=center] at (1.25, 2.25) {"ZSZ" \\ Code \cite{mfmt-fwkg}};
          % Abelian 2BGA
         \draw[thick, fill=white!10] (0,-2.50) ellipse (3.40cm and 4.0cm);
         \node[align=center] at (-1.0, 0.65) {Abelian \\ 2BGA};
          % BB
         \draw[thick, fill=white!10] (1.15,-2.50) ellipse (2.25cm and 3.5cm);
         \node[align=center] at (0.9, 0.1) {BB \\ Code \cite{Bravyi2024}};
         %Generalized Toric
         \draw[thick, fill=white!10] (2.00,-1.55) ellipse (1.25cm and 1.50cm);
         \node[align=center] at (1.90, -1.00) {Generalized \\ toric\\ code \cite{rmy6-9n89}};
         %Cyclic HGP
         \draw[thick] (2.00,-3.40) ellipse (1.25cm and 1.50cm);
         \node[align=center] at (1.90, -3.7) {Cyclic \\ HGP \\ \cite{aydin2026cyclichypergraphproductcode}};
         %GB Codes
         \draw[thick, rotate around={20:(-1.25,-4.50)}] (-0.80,-4.50) ellipse (1.25cm and 1.25cm);
         \node[align=center] at (-1.30, -4.5) {GB \\Code \cite{PhysRevA.88.012311}};
         %TB
         \draw[thick, fill=white!10] (-2.10,-1.50) ellipse (0.95cm and 1.4cm);
         \node[align=center] at (-2.1, -0.65) {TB \\ code \\ \cite{ll5p-z88p}};
         %TB
         \draw[thick, fill=white!10] (-2.10,-2.15) ellipse (0.75cm and 0.75cm);
         \node[align=center] at (-2.1, -2.05) {Cubic \\ code \\ \cite{Haah2011}};
     \end{tikzpicture}
     \caption{Various quantum LDPC codes and their relationship with each other.}
     \label{fig:CodeFamilies}
 \end{figure*}
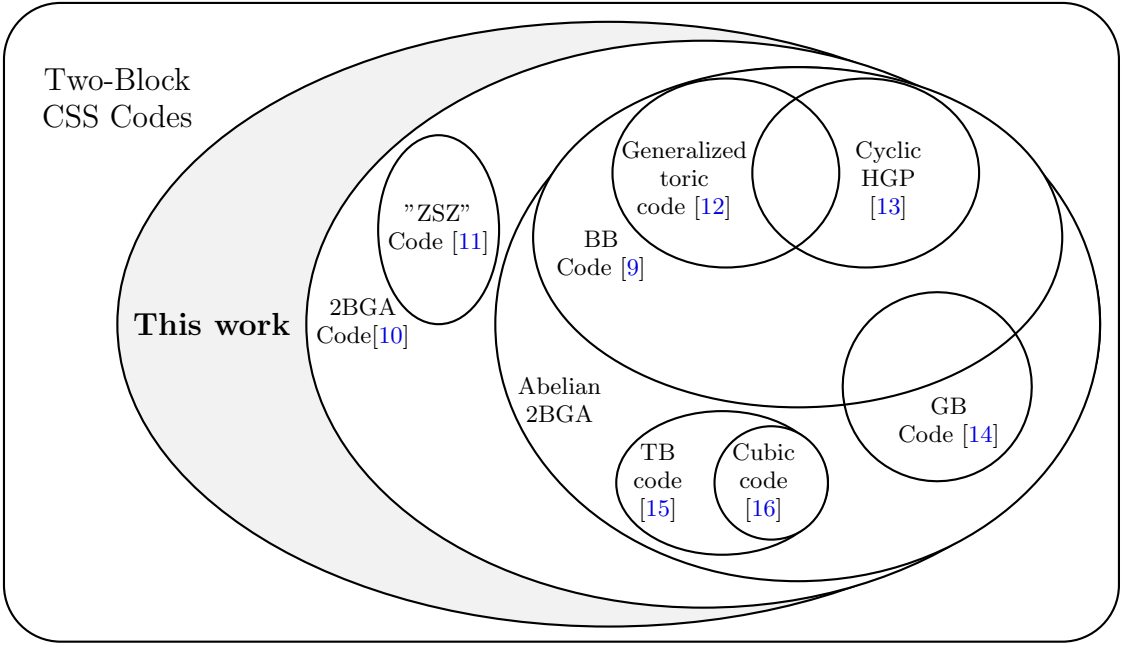

\begin{table*}[htbp]
    \centering 
    \caption{Mapping existing quantum LDPC codes to our generalized framework. }
    \label{tab:code_families_relationship}
    \begin{tabular}{@{}ll@{\quad}ll@{}}
        \toprule
        \textbf{Code Family} & \textbf{Subgroup Condition} & \textbf{Quotient Group} ($G/H$) & \textbf{Reference} \\
        \midrule
        \textbf{Coset-based codes} & \textbf{Arbitrary} $H \leq G$ & \textbf{Coset space} $G/H$ & this work \\
        \midrule
        2BGA & Normal $H$ & General finite group & Lin-Pryadko \cite{Pryadko2BGA} \\
        ``ZSZ'' Code & Normal $H$ & $\mathbb{Z}_l \rtimes \mathbb{Z}_m$ & Guo e.a. \cite{mfmt-fwkg} \\
        Abelian 2BGA & Normal $H $ & Abelian finite group & Lin-Pryadko \cite{Pryadko2BGA} \\
        Trivariate Bicycle (TB) & Normal $H $ & $\mathbb{Z}_k \times \mathbb{Z}_l \times \mathbb{Z}_m$ & Voss e.a. \cite{ll5p-z88p} \\
        Cubic Code & Normal $H$ & $\mathbb{Z}_l \times \mathbb{Z}_l \times \mathbb{Z}_l$ & Haah \cite{Haah2011} \\
        Bivariate Bicycle (BB) & Normal $H $ & $\mathbb{Z}_l \times \mathbb{Z}_m$ & Bravyi e.a. \cite{Bravyi2024} \\
        Generalized Toric & Normal $H$ & $\mathbb{Z}_l \times \mathbb{Z}_m$ & Liang e.a. \cite{rmy6-9n89} \\
        Cyclic HGP & Normal $H $ & $\mathbb{Z}_l \times \mathbb{Z}_m$ & Aydin-Delfosse-Tham\cite{aydin2026cyclichypergraphproductcode} \\
        Generalized Bicycle (GB) & Normal $H$ & $\mathbb{Z}_r$ & Kovalev-Pryadko \cite{PhysRevA.88.012311} \\
        \bottomrule
    \end{tabular}
\end{table*}

BB codes belong to a broader family of codes known as two-block group algebra (2BGA) codes \cite{Pryadko2BGA}. These codes are constructed using permutation matrices that represent the left and right regular actions of a finite group, which is not limited to being abelian. Since the left and right regular actions commute with each other, the CSS orthogonality condition is automatically satisfied owing to the general construction method known as two-block CSS codes \cite{PhysRevA.88.012311}. 

One of the primary reasons for the interest in 2BGA codes is the well-defined structure that they imply for the resulting Tanner graph. Namely, 
the edges of the graph that connect the qubits correspond to the actions of the underlying finite group, which makes it possible to design an efficient CNOT gate schedule that complies with the constraints imposed by the hardware architecture. For instance, the structure of the underlying group $ \mathbb{Z}_l \times \mathbb{Z}_m$, along with the structure of the constituent matrices, enables one to map BB codes onto a 2D bi-planar layout, which is highly desirable for superconducting hardware architectures. Furthermore, it has been demonstrated that BB codes are suitable for implementation in trapped-ion architectures \cite{Ye2025quantumerror,tham2025distributedfaulttolerantquantummemories} and cold-atom platforms \cite{Wang2026coprimebivariate}.    

Another example of 2BGA codes is a class of hypergraph product (HGP) codes known as cyclic HGP codes \cite{aydin2026cyclichypergraphproductcode}, which form
a special case of BB codes and admit an efficient layout for trapped-ion platforms. Generalized bicycle (GB) codes \cite{PhysRevA.88.012311} 
form yet another example of 2BGA codes that relies on the group $ \mathbb{Z}_r$ for its construction. Note that GB codes are equivalent to BB codes when $l$ and $m$ are coprime, due to the group isomorphism $\mathbb{Z}_{lm} \cong \mathbb{Z}_l \times \mathbb{Z}_m$. Furthermore, Ref.~\cite{Panteleev2021degeneratequantum}
introduced several examples of finite-length GB codes that feature both good parameters and high thresholds. More recently, the authors of \cite{tripier2026faulttolerantquantumcomputingtrapped} proposed a blueprint for a fault-tolerant trapped-ion quantum computer architecture. The authors defined a \textit{three-ring} framework based on the ``moving qubits'' model (a fully connected trapped-ion architecture where qubits can physically move) and introduced an efficient quantum memory utilizing 2BGA codes, including the cyclic HGP, BB, and GB code families.

2BGA codes based on non-abelian groups have also been investigated for the construction of self-correcting quantum memories in neutral-atom arrays \cite{mfmt-fwkg}. In the same work, the authors specifically studied 2BGA codes with the underlying group formed as a semidirect product of two cyclic groups $ \mathbb{Z}_l \rtimes \mathbb{Z}_m$. Several other code families in the class of 2BGA codes are trivariate bicycle codes \cite{ll5p-z88p} with the underlying group $ \mathbb{Z}_k \times \mathbb{Z}_l \times \mathbb{Z}_m$ and Haah's cubic code \cite{Haah2011}, which is a special case of the former relying on the underlying group $ \mathbb{Z}_l \times \mathbb{Z}_l \times \mathbb{Z}_l$.

These examples are listed in Table~\ref{tab:code_families_relationship} and their relationships are visualized in 
Figure \ref{fig:CodeFamilies}. Collectively, they demonstrate the potential of 2BGA codes for near-term fault-tolerant quantum memories. An extensive analysis of 2BGA codes for both abelian and non-abelian groups, including a computer search for codes with underlying groups up to a certain order, has been conducted 
by Lin and Pryadko in Ref.~\cite{Pryadko2BGA}. 

In this work, we introduce a general family of quantum LDPC codes that contains and extends 2BGA codes and hence also their subfamilies listed in Table~\ref{tab:code_families_relationship}. Specifically, we present a construction of two-block CSS codes based on group actions on cosets rather than the regular group actions utilized in 2BGA codes. While for certain 
group-subgroup pairs, our approach is equivalent to the 2BGA construction, generally speaking, it yields new, previously unknown quantum LDPC codes. While there is a limited number of finite groups of a given order, there are many more group-subgroup pairs that yield a coset space of that same size. Consequently, for a fixed code length, our construction provides a much wider variety of Tanner graph structures and an expanded search space for discovering new quantum codes. 

Aided by computer search, we identify several new quantum LDPC codes with competitive parameters (see Tables \ref{tab:group_codes} and \ref{tab:group_codes_additional}). We also design an explicit, maximally packed syndrome extraction circuit for our codes. Any weight-$w$ code from our family admits a depth-$(w+2)$ syndrome extraction circuit, including state initialization and measurement operations.  We further analyze the structure of the Tanner graph of our codes and provide an upper bound for its thickness. By construction, our results apply to all 2BGA codes, such as BB and GB codes. In simulations, we observe that our codes perform comparably to BB codes under circuit-level noise (see Table \ref{tab:summarylogicalerrorrates}), while extending the set of possible code parameters.

As another set of results, we present a group-theoretic framework for generating larger quantum LDPC codes by lifting a small base 2BGA code. In particular, this construction recovers and generalizes the results of Symons et al. \cite{symons2025sequencesbivariatebicyclecodes}, in which the authors studied families of BB codes from covering graphs.  Additionally, we provide several notable examples of quantum LDPC codes that are graph-based covers of previously known 2BGA codes.

The paper is organized as follows: In Section \ref{sec:GroupTheory}, we provide a brief overview of group actions on cosets. We describe our code construction in Section \ref{sec:CodeConstruction} and introduce our search strategy for finding new codes in Section \ref{sec:CodeSearch}. Section \ref{sec:SC_Circuit} presents an explicit syndrome extraction algorithm for our code family alongside the results of our circuit-level memory experiments. Finally, we conclude by describing how to generate a sequence of our codes using covering graphs in Section \ref{sec:CoverGraphs}.

\section{\label{sec:GroupTheory} Groups Acting on Cosets}
Our code construction relies on the action of a finite group $G$ on the cosets of its subgroup $H$. In this section, we list
the main properties of this action. For a subgroup $H\leq G$, let $N_G(H):=\{g\in G : gH=Hg\}$ the normalizer of $H$ in $G$. Let $S_{G/H}$ be the group of permutations of the left cosets of $H$ in $G$. We define the standard left action of $G$ and right action of $N_G(H)$ on the left cosets as follows:
\begin{align}\label{eq:defL}
    L : G \rightarrow S_{G/H},
\end{align}
where
\begin{align*}
    L(g) : G/H &\rightarrow G/H\\
    xH &\mapsto (gx)H.
\end{align*}
Similarly,
\begin{align}\label{eq:defR}
    R : N_G(H) \rightarrow S_{G/H},
\end{align}
where
\begin{align*}
    R(g) : G/H &\rightarrow G/H\\
    xH &\mapsto (xg)H.
\end{align*}
In this section we list standard properties of these actions \cite[Sec.4.2]{Dummit2004}, including their proofs in 
Appendix \ref{appendix:additionalproofs} for completeness. Recall that a group homomorphism (anti-homomorphism) is a map $G_1\stackrel{\phi}\to G_2$ such that $\phi(g_1g_2)=\phi(g_1)\phi(g_2)$ (resp., $\phi(g_1g_2)=\phi(g_2)\phi(g_1)$).
In the particular case when $G_2$ is a permutation group of a finite set, $\phi$ is called a permutation (anti)-representation of $G_1$.
\begin{proposition}\label{prop:LRhomomorpshim}
$L$ is a homomorphism and $R$ is an anti-homomorphism.
\end{proposition}
Crucially, the left action of $G$ and the right action of $N_G(H)$ on the left cosets of $H$ in $G$ commute with each other.  Using this property in our code construction, we immediately establish the orthogonality requirement imposed by the CSS code family. We formalize this property in the following proposition:
\begin{proposition}\label{prop:commutativity}
    Let $G$ be a finite group and let $H\leq G$ be a subgroup. Let $L$ and $R$ be the maps defined in \eqref{eq:defL} and \eqref{eq:defR}. Let $g_1\in G$ and $g_2\in N_G(H)$. Then,
    \begin{align*}
        L(g_1)\circ R(g_2) = R(g_2) \circ L(g_1).
    \end{align*}
\end{proposition}
Recall that the {\em core} of $H$ in $G$ is the largest normal subgroup of $G$ contained in $H$:
\begin{align*}
    \operatorname{Core}_G(H) = \bigcap_{g \in G} gHg^{-1}.
\end{align*}
The number of distinct left actions on the left cosets is given by $|G|/|\operatorname{Core}_G(H)|$, while the number of distinct right actions is $|N_G(H)|/|H|$. This follows directly from the following proposition:

\begin{proposition}\label{prop:ImageLandR}
    Let $G$ be a finite group and $H \leq G$ be a subgroup. Then
  $\ker L\cong \operatorname{Core}_G(H)$ and $\ker R\cong H$. Consequently,
    \begin{align*}
        \operatorname{Im}(L) \cong G/\operatorname{Core}_G(H) \text{ and } \operatorname{Im}(R) \cong N_G(H)/H.
    \end{align*}
\end{proposition}
We will make extensive use of the permutation matrices that encode this action. 
\begin{definition}\label{def:matrices}
Let $G$ be a finite group and $H \leq G$ be a subgroup. Let $G/H = \{x_1H, x_2H, \ldots, x_mH\}$, where $(x_i)_i$ is a fixed set of coset representatives and $m:=[G:H]$ is the index of $H$ in $G$. For each $g \in G$, define an $m \times m$ permutation matrix $\bfL(g)$ by setting
    \begin{align*}
        \bfL(g)_{i,j} = 
            \1_{\{L(g)(x_jH) = x_iH\}}.
    \end{align*}
    Similarly, for each $g \in N_G(H)$, define the $m \times m$ permutation matrix $\bfR(g)$ as:
    \begin{align*}
        \bfR(g)_{i,j} =  \1_{\{R(g)(x_jH) = x_iH\}}.
    \end{align*}
\end{definition}

Switching from permutations to permutation matrices, we also use a new notation for the permutation (anti-) representation, writing 
   \begin{align*}
        \bfL : G \rightarrow GL_m(\mathbb{C})
    \end{align*}
and 
   \begin{align*}
        \bfR : N_G(H) \rightarrow GL_m(\mathbb{C}).
    \end{align*}
for the left and right action on the cosets.

The following example illustrates the construction of permutation matrices representing left and right group actions on cosets.
\begin{example}\label{example: D1}
    Let $G=\langle r,s\mid r^4=s^2=(sr)^2=1\rangle$ be the dihedral group of order $8$ and $H=\langle s\rangle=\{e,s\}$ be a subgroup. Let us fix the representatives as $x_1=e, x_2=r, x_3=r^2$ and write the cosets as follows:
    \begin{align*}
        G/H=\{ H, rH, r^2H,r^3H \}
    \end{align*}
Since $\operatorname{Core}_G(H)=\{e\}$, we have $G/\operatorname{Core_G(H)}\cong G$. 
By Proposition \ref{prop:ImageLandR}, we have $|\operatorname{Im}(\bfL)|=8$, so there are $8$ unique permutation matrices in the image of $\bfL$. Let us illustrate how these entries are computed from Definition \ref{def:matrices}. 
Since $$L(r)(x_1H) = r(eH) = rH = x_2H,$$ we find that $\bfL(r)_{2,1}=1$. Similarly, for the generator $s$, we have 
  $$L(s)(x_2H) = s(rH) = srH = r^3sH = r^3H = x_4H,$$ which gives $\bfL(s)_{4,2} = 1$. Continuing in the same manner, we obtain
    \begin{align*}
        \bfL(r)=\begin{bmatrix}
            0 & 0 & 0 & 1\\
            1 & 0 & 0 & 0\\
            0 & 1 & 0 & 0 \\
            0 & 0 & 1 & 0
        \end{bmatrix},\quad  \bfL(s)=\begin{bmatrix}
            1 & 0 & 0 & 0\\
            0 & 0 & 0 & 1\\
            0 & 0 & 1 & 0 \\
            0 & 1 & 0 & 0
        \end{bmatrix}
    \end{align*}
    The remaining left action matrices can be constructed using the property of the representation: $$\bfL(e)=\bfL(r)^4=\bfL(s)^2=\bfI, \quad \bfL(rs)=\bfL(r)\bfL(s),$$ etc. 
 
    Next, the normalizer is $N_G(H)=\{e,s,r^2,r^2s\}$ and we find
    \begin{align*}
        {N_G(H)}/{H}=\{H, r^2H\}.
    \end{align*}
 Again, by Proposition \ref{prop:ImageLandR}, there are 2 unique right action permutation matrices, namely, $\bfR(e)=\bfR(s)$ and $\bfR(r^2)=\bfR(r^2s)$. Explicitly,
    \begin{align*}
        \bfR(e)=\bfR(s)=\bfI; \quad
        \bfR(r^2)=\bfR(r^2s)=\begin{bmatrix}
    0 & 0 & 1 & 0\\
    0 & 0 & 0 & 1\\
    1 & 0 & 0 & 0 \\
    0 & 1 & 0 & 0
\end{bmatrix}
    \end{align*}$\triangleleft$
\end{example}

For a finite field $F$ and a finite group $G$, the group algebra $F[G]$ is a set of formal linear combinations
\begin{align*}
    F[G] = \Big\{\sum_{g\in G} \lambda_gg : \lambda_g\in F \Big\}.
\end{align*}
Elements $\ah=\sum_{g\in G}\alpha_gg$ and $\bh=\sum_{g\in G}\beta_gg$ can be added and multiplied 
according to 
\begin{align*}
    \ah + \bh=  \sum_{g\in G}(\alpha_g+\beta_g)g\\
    \ah*\bh =  \sum_{g,h\in G}(\alpha_h\beta_{h^{-1}g})g, 
\end{align*}

Let $H\le G$ be a subgroup and let 
  \begin{align*}
\ah=\sum_{g\in G} \alpha_gg, \quad
\bh=\sum_{g\in N_G(H)} \beta_gg.
\end{align*}
The maps $\bfL$ and $\bfR$ define a natural action on $F[G]$ anf $F[N_G(H)]$, respectively, according to
\begin{align*}
    \bfL(\ah) = \sum_{g\in G}\alpha_g\bfL(g), \quad
    \bfR(\bh) = \sum_{g\in N_G(H)}\beta_g\bfR(g).
\end{align*}
Note that $\bfL(\ah_1*\ah_2)=\bfL(\ah_1)\bfL(\ah_2)$, $\bfR(\bh_1*\bh_2)=\bfR(\bh_2)\bfR(\bh_1)$, and 
$[\bfL(\ah),\bfR(\bh)]=0$ by Proposition \ref{prop:commutativity}.

\section{\label{sec:CodeConstruction} Code Construction}
In this section, we present our code construction, which is based on the framework of two-block CSS codes. Originally introduced 
by Kovalev and Pryadko \cite{PhysRevA.88.012311}, these codes can be viewed as a generalization of an early paper on bicycle codes, proposed by MacKay \textit{et al.}~\cite{MacKayBicyle}. A {\em two-block CSS code} is defined by parity-check matrices of the form
\begin{align*}
    H_X = [A \mid B] \quad \text{and} \quad H_Z = [-B^T \mid A^T],
\end{align*}
where $A$ and $B$ are square commuting matrices over a finite field. Using the matrices $H_X$ and $H_Z$ in this form
enables one to fulfill effortlessly the CSS commutativity condition: indeed, $H_X H_Z^T = -AB + BA = 0$. Our construction belongs to two-block CSS codes and can be described as follows:
\begin{construction}\label{cons:Generalized2BGA}
Let $F$ be a finite field, let $G$ be a finite group, let $H\le G$, and let $N=N_G(H)$ be the normalizer of $H$ in $G$. For a given $\ah\in F[G]$ and $\bh\in F[N]$,  define a two-block CSS code $Q_G^H(\ah,\bh)$ via the parity-check matrices
    \begin{align*}
        H_X=[\bfL(\ah) \mid \bfR(\bh)]\quad \text{and}\quad H_Z=[-\bfR(\bh)^T\mid \bfL(\ah)^T].
    \end{align*}
We call the code $Q_G^H$ {\em binary} if $F={\mathbb F}_2$.
\end{construction}

Before analyzing this construction, we comment on its relation to existing work. 
    Two-block group algebra (2BGA) codes \cite{Pryadko2BGA} form a large family of qLDPC codes, including many well-known families, such as Generalized Bicycle (GB) codes \cite{PhysRevA.88.012311} and Bivariate Bicycle (BB) codes \cite{Bravyi2024}. The code family 
$Q_G^H(\ah,\bh)$ generalizes the 2BGA construction: indeed if $H\triangleleft G$ is normal, then $\operatorname{Core}_G(H)= H$. By Proposition \ref{prop:ImageLandR}, in this case, the matrices $\bfL$ are the permutation matrices representing the left regular action of the quotient group $G/H$ on itself. Similarly, because $N_G(H) = G$ when $H$ is a normal subgroup, the  matrices $\bfR$ describe the right regular action of $G/H$, and 
thus, the code $Q_G^H(\ah,\bh)$ reduces to the {2BGA} code defined over the group $G/H$. Further specializing $G$, we obtain most families of
BB-related codes known in the literature. The precise reductions are summarized in Table \ref{tab:code_families_relationship}, and 
Figure \ref{fig:CodeFamilies} presents a visual illustration of the interrelations between these code families.

    Let us define the support sets of the group algebra elements as
    \begin{align*}
        &\supp(\ah) := \{g \in G : \alpha_g \neq 0\}, \\
        &\supp(\bh) := \{g \in N_G(H) : \beta_g \neq 0\},
    \end{align*}
    and refer to the number of non-zero terms as their respective \textit{weights} denoted by $w_{\ah} = |\operatorname{supp}(\ah)|$ and $w_{\bh} = |\operatorname{supp}(\bh)|$.  
In the following proposition, we record the parameters provides an upper bound on the stabilizer weights of $Q_G^H(\ah,\bh)$.

\begin{proposition}\label{prop:CodeProperties}
    Let $Q_G^H(\ah,\bh)$ be the code defined in Construction \ref{cons:Generalized2BGA} and let $m=[G:H]$  Then $Q_G^H(\ah,\bh)$ has the parameters
    \begin{align*}
        &n = 2m, \\
        &k = n - \rank(H_X) - \rank(H_Z), \\
        &d = \min(d_X, d_Z),
    \end{align*}
    where the $X$- and $Z$-distances are given by
    \begin{align*}
        &d_X = \min\{\wt(v) : v \in \ker H_X \setminus \operatorname{Im} H_Z^T\}, \\ 
        &d_Z = \min\{\wt(v) : v \in \ker H_Z \setminus \operatorname{Im} H_X^T\}.
    \end{align*}
The weight of the stabilizer generators of the code $Q_G^H(\ah,\bh)$ satisfies  $w \leq w_{\ah} + w_{\bh}$.
\end{proposition}

\begin{proof}
 The matrices $\bfL$ and $\bfR$ each have dimension $m\times m$, making the code length $n$ equal to $2m$. The dimension and distance follow
 since the code $Q_G^H(\ah,\bh)$ is a CSS code.  The statement about the stabilizer weight is immediate.
\end{proof}
We say that the code $Q_G^H(\ah,\bh)$ has \textit{regular stabilizer weight} if the weights of all stabilizer generators in the matrices $H_X$ and $H_Z$ satisfy $w = w_{\ah} + w_{\bh}$.

\begin{remark}
    By Proposition \ref{prop:ImageLandR}, choosing the support of the group algebra elements $\ah$ and $\bh$ from the representatives of the cosets $G/\operatorname{Core}_G(H)$ and $N_G(H)/H$ ensures that the underlying permutation matrices composing the sums $\bfL(\ah)$ and $\bfR(\bh)$ remain completely distinct 
\end{remark}

\subsection{Planar connectivity and code thickness}
In addition to the code parameters and noise tolerance, implementations call for CSS codes with planar or close-to-planar Tanner graphs, which
simplify connectivity and design of syndrome extraction circuits. Code thickness, which, roughly speaking, measures how close the Tanner graph is to being planar, is one way to quantify this requirement. In this section, we analyze the thickness of our code construction.

To define the Tanner graph of an $[[n, k, d]]$ CSS code $\mathcal{Q}$, recall that it is described by a set of data qubits $V_D=\{q_1,q_2,\ldots,q_n\}$, a set of $X$-checks $V_X=\{X_1,X_2,\ldots, X_{r_X}\}$, and a set of $Z$-checks $V_Z=\{ Z_1,Z_2,\ldots,Z_{r_Z} \}$. Each $X$- or $Z$-check corresponds to an $X$- or $Z$-stabilizer generator of the code $\mathcal{Q}$. If the stabilizer generators are independent, then $n=k+r_X+r_Z$.
\begin{definition}[Tanner graph]
The $X$-{\em Tanner graph} $\tau_X=(V_D\cup V_X, E_X)$ of the code $\mathcal{Q}$ is the bipartite graph whose vertex set is $V_D\cup V_X$ and the connections $E_X$ are given by the matrix $H_X$. Similarly, the $Z$-{\em Tanner graph} has vertices $V_D\cup V_Z$ and connections $E_Z$ defined by the matrix $H_Z$.
Finally, the full {\em Tanner graph} of the code $\mathcal{Q}$ is defined as the union $\tau=\tau_X\cup\tau_Z$, with the vertex set $V = V_D \cup V_X \cup V_Z$ and the edge set $E = E_X \cup E_Z$. 
\end{definition}

The planarity of the Tanner graph of a quantum code can be an important consideration for certain hardware architectures, since crossing connections between qubits can induce additional crosstalk noise. Therefore, planar codes, such as surface codes, are favorable due to these practical considerations. When planarity is not achievable, a solution can be codes whose Tanner graphs can be decomposed into a small number of planar subgraphs. 
This approach can be quantified using the concept called graph thickness. Formally, the thickness of a graph, denoted $\theta$ below, is defined as the minimum number of planar subgraphs into which its edges can be partitioned. 

CSS codes with low thickness were recently studied in \cite{PhysRevLett.129.050504}. By a classic result from graph theory \cite{HALTON1991219}, CSS codes whose Tanner graph has maximal degree $\delta$,
can be implemented in $\lceil\delta/2\rceil$ planar layers. Using this result for the code $Q_G^H(\ah,\bh)$, we obtain a
bound $\theta \leq \lceil(w_{\ah}+w_{\bh})/2\rceil$. Our main result in this section states that, relying on the highly structured Tanner graph of $Q_G^H(\ah,\bh)$, this bound can be improved as follows.

\begin{lemma}\label{lemma:thickness}
    Let $Q_G^H(\ah,\bh)$ be the code of length $n$ defined in Construction \ref{cons:Generalized2BGA} constructed over the binary field. 
     Assume the code has regular stabilizer weight of $w=w_{\ah}+w_{\bh}$,
    where $w_{\ah}$ and $w_{\bh}$ are the weights of the group algebra elements $\ah \in \mathbb{F}_2[G]$ and $\bh\in\mathbb{F}_2[N_G(H)]$. \ref{prop:CodeProperties}
     Then the thickenss of the Tanner graph of $Q_G^H(\ah,\bh)$ satisfies the inequalities
    \begin{align*}
        \Big\lceil \frac{nw}{4(n-1)}\Big\rceil\leq \theta \leq \max\left(\left\lceil{\frac{w_{\ah}}{2}}\right\rceil, \left\lceil{\frac{w_{\bh}}{2}}\right\rceil, \left\lceil \frac{w}{3}\right\rceil \right).
    \end{align*}
In particular, if $w_{\ah} = w_{\bh}$, then the thickness of a code $Q_G^H(\ah,\bh)$ of length $n\geq 4(\lfloor w/4\rfloor +1)/(4-(w\bmod 4))$ satisfies
    \begin{align*}
        \left\lfloor\frac{w}{4}\right\rfloor + 1 \leq \theta \leq \left\lceil \frac{w}{3} \right\rceil.
    \end{align*}
\end{lemma}

\begin{proof}
    The proof for the upper bound is a straightforward generalization of Lemma 2 in \cite{Bravyi2024}. Let $\mathcal{T} = \mathcal{T}_1 \cup \mathcal{T}_2 \cup \ldots \cup \mathcal{T}_\theta$ be an edge-partition of the Tanner graph of $Q_G^H(\ah,\bh)$ into planar subgraphs
    Let each subgraph $\mathcal{T}_i$ be a bipartite graph defined by the parity-check matrices $H_X^{(i)}$ and $H_Z^{(i)}$. We claim that any such subgraph defined by the matrices of the form
    \begin{align}\label{eq:layerMatrices1}
        &H_X^{(i)}=[\bfL(g_r)+\bfL(g_s)\mid \bfR(g_t)],\\
        &H_Z^{(i)}=[\bfR(g_t^{-1})\mid \bfL(g_r^{-1})+\bfL(g_s^{-1})]
    \end{align}
    or
    \begin{align}\label{eq:layerMatrices2}
        &H_X^{(i)}=[\bfL(g_r)\mid \bfR(g_t)+\bfR(g_q)],\\
        &H_Z^{(i)}=[\bfR(g_t^{-1})+\bfR(g_q^{-1})\mid \bfL(g_s^{-1})]
    \end{align}
    is planar for every $g_r,g_s\in G$ and $g_q,g_t\in N_G(H)$. For the case of BB codes, this is proved in 
    \cite{Bravyi2024}; see the proof of their Lemma 2,
   where $G=\mathbb{Z}_l\times \mathbb{Z}_m$ and $H=\{e\}$. However, their result affords a straightforward 
   extension to any $Q_G^H(\ah,\bh)$ code. Namely, following the argument in \cite{Bravyi2024}, one can show that each subgraph $\mathcal{T}_i$ is composed of disconnected \textit{wheel graphs.} Each wheel graph consists of an inner cycle connected to an outer cycle of the same length. The length of both cycles is $l$, which is an integer given by $(g_rg_s^{-1})^l=I$ (or $(g_tg_q^{-1})^l=I$). There are $l$ radial edges equally spaced from each other. A detailed description of the wheel graph is given in the Appendix \ref{appendix:wheel graph}. Proposition \ref{prop:commutativity} implies that 
    such a graph is well defined, and since each individual wheel graph is planar, the subgraph $\mathcal{T}_i$ corresponding to the $i$-th layer must be planar.

    Therefore, an upper bound on the thickness can be found by partitioning the $\bfL$ and $\bfR$ matrices into different layers obeying the forms in \eqref{eq:layerMatrices1} and \eqref{eq:layerMatrices2}. Let $w_{\ah}^{(i)}$ and $w_{\bh}^{(i)}$ be the number of $\bfL$ and $\bfR$ matrices assigned to the $i$-th layer. For the subgraphs to take the planar forms described above, we must simultaneously satisfy $w_{\ah}^{(i)}, w_{\bh}^{(i)} \leq 2$ and $w_{\ah}^{(i)} + w_{\bh}^{(i)} \leq 3$ for each layer. Partitioning the $w_{\ah}$ total $\bfL$ matrices and $w_{\bh}$ total $\bfR$ matrices under these constraints directly yields the upper bound stated in the lemma.

    The proof for the lower bound follows from the foundational corollary of Euler's formula for bipartite graphs: Any simple, planar bipartite graph with $V\geq 3$ vertices must satisfy $E\leq 2V-4$, where $E$ is the number of edges \cite[Thm.6.1.23]{west2002}. This corollary yields a lower bound on the thickness, that is
    \begin{align*}
        \left\lceil \frac{E}{2V-4} \right\rceil\leq \theta.
    \end{align*}
    Note that we have $n/2$ Z checks, $n/2$ X checks, and $n$ data qubits, making the total number of vertices $V=2n$. Since we assume the code attains the exact regular stabilizer weight $w=w_{\ah}+w_{\bh}$, we have exactly $E=nw$ edges, yielding the bound
    \begin{align*}
        \left\lceil \left(\frac{n}{n-1}\right)\frac{w}{4}\right\rceil\leq \theta. 
    \end{align*}
    Note that if the code length satisfies $n\geq 4(\lfloor w/4\rfloor +1)/(4-(w\bmod 4))$, then
    \begin{align*}
        \left\lceil \left(\frac{n}{n-1}\right)\frac{w}{4}\right\rceil = \left\lfloor \frac{w}{4} \right\rfloor + 1,
    \end{align*}
    completing the proof.
\end{proof}

Note that any weight-$6$ code $Q_G^H(\hat a,\hat b)$ with $w_{\ah}=w_{\bh}=3$
has thickness exactly $\theta=2$. Indeed, the upper bound in
Lemma~\ref{lemma:thickness} gives $\theta\le 2$, while the Euler lower bound
gives $\theta\ge 2$ for all $n\ge 4$. This condition is automatically
satisfied here: since the block $\bf R(\hat b)$ has row weight $3$ and has
 $[G:H]$ columns, we must have $[G:H]\ge 3$, and hence the code length
\(n=2[G:H]\) satisfies \(n\ge 6\).
Similarly, any weight-$8$ $Q_G^H(\ah,\bh)$ code with $w_{\ah}=w_{\bh}=4$ must have a thickness of exactly $\theta=3$ by Lemma \ref{lemma:thickness}.

\section{\label{sec:CodeSearch} Code Search}
To find codes with good parameters, we performed a computer search in the set of quantum LDPC codes introduced in this paper, examining
binary codes with component weights $w_\ah=w_\bh=3$ and $4$. We limited ourselves to non-Abelian groups and their non-normal subgroups to expand our search space beyond the 2BGA code family. To improve the efficiency of the search process, it is important to eliminate equivalent codes in advance,
which raises the question of characterizing equivalences between the codes in the family $Q_G^H(\ah,\bh)$. 
\subsection{Code Equivalences}

We say that two quantum error correction codes are equivalent if one is obtained from the other by performing some set of single-qubit unitaries and qubit permutations \cite{Gottesman2024}. Such equivalent codes share the same logical parameters. In the context of CSS codes, code equivalence corresponds to applying certain column operations to the parity-check matrices, i.e., permutations and non-zero scalar multiplications. 

Identifying equivalent codes saves computational resources and enables one to examine larger code sets, increasing the chances of finding
good codes. For this purpose, the authors of \cite{Pryadko2BGA} studied equivalent codes for 2BGA codes (see Theorem 6 in \cite{Pryadko2BGA}). As discussed earlier, their results apply to the codes $Q_G^H(\ah,\bh)$ when $H$ is a normal subgroup of $G$. 
Since we focus on non-normal subgroups, we modify the argument of \cite{Pryadko2BGA} to show that acting on the code $Q_G^H(\ah,\bh)$ by group elements and 
constant multiplication results in an equivalent code.

\begin{lemma}\label{lemma:equivalence1}
    Let $Q_G^H(\ah,\bh)$ be a CSS code defined in Construction \ref{cons:Generalized2BGA}. Let $g_a\in G$, $g_b\in N_G(H)$, and $\alpha,\beta\in F\setminus \{0\}$.
    Then the code $Q_G^H(\ah,\bh)$ is equivalent to the code $Q_G^H(\ah * (\alpha g_a), (\beta g_b)*\bh)$.
\end{lemma}
\begin{proof}
    Let $H_X,H_Z$ be the parity-check matrices of the code $Q_G^H(\ah,\bh)$ and $H_X^\prime,H_Z^\prime$ be the parity-check matrices of the code $Q_G^H(\ah * (\alpha g_a), (\beta g_b)*\bh)$. We need to show that $H_X^\prime$ and $H_Z^\prime$ can be obtained by permuting and scaling the columns and/or row operations on $H_X$ and $H_Z$, respectively. Let us define the matrix
    \begin{align*}
        P_X=\begin{bmatrix}
            \bfL(\alpha g_a) & \bf0\\
            \bf0 & \bfR(\beta g_b)
        \end{bmatrix}.
    \end{align*}
Using the relations 
\begin{gather*}
\bfL(\ah*(\alpha g_a))=\bfL(\ah)\bfL(\alpha g_a)\\
\bfR((\beta g_b)*\bh)=\bfR(\bh)\bfR(\beta g_b),
\end{gather*}
we observe that $H_X^\prime=H_XP_X$.
 In other words, $H_X^\prime$ is obtained by applying column operations (permutations and scaling) to $H_X$. Applying the transform $P_X$ to the data qubits requires us to perform a corresponding transformation $P_Z=(P_X^{-1})^T$ of the $Z$ stabilizers to preserve the CSS orthogonality property.
Now, let us define the invertible row-operation matrix
    \begin{align*}
        M=\alpha\beta\bfR(g_b)^T\bfL(g_a)^T.
    \end{align*}
Performing matrix multiplication, we observe that $H_Z^\prime=M(H_ZP_Z)$,
{This matrix can be obtained by applying the corresponding column operations to $H_Z$, followed by a set of valid row operations. } This completes the proof.
\end{proof}
Lemma \ref{lemma:equivalence1} suggests the possibility of reducing the search space. Namely, suppose we search for codes with a fixed weight $w=w_{\ah}+w_{\bh}$, where $w_{\ah}$ and $w_{\bh}$ are defined as in Proposition \ref{prop:CodeProperties}. Choose $w_{\ah}$ distinct elements from $G$ that are coset representatives for the quotient group $G/\operatorname{Core}_G(H)$ and $w_{\bh}$ distinct elements from $N_G(H)$ that are coset representatives for $N_G(H)/H$. Lemma \ref{lemma:equivalence1} implies that we can fix one element in each of  $\ah$ and $\bh$ without loss of generality to the equivalence claim.
Therefore, it is sufficient to choose $w_{\ah}-1$ (resp. $w_{\bh}-1$) elements instead of $w_{\ah}$ (resp. $w_{\bh}$), which significantly reduces the size of the search space depending on the group size. 

Additionally, reducing the size of the search space can be accomplished by excluding certain redundant subgroups. In practice, computer algebra systems such as GAP return exhaustive lists of subgroups that include many conjugate copies. Searching over all such conjugates would needlessly duplicate the same code structures, increasing the computational load. To prevent this, in the following lemma, we show that 
codes constructed using conjugate subgroups are equivalent to each other. To state the lemma, we first relate conjugation and normalizers.
Let $H_1< G$ be a group-subgroup pair and let $H_2=g^{-1}H_1g$ 
    be a subgroup conjugate to $H_1$ for some $g\in G$. Then $N_G(H_2)=g^{-1}N_G(H_1)g$ and so $g^{-1}*\bh*g\in F[N_G(H_2)]$.

\begin{lemma}\label{lemma:ConjugateSubgroup}
Let $H_1<G, H_2=g^{-1}H_1 g$ for $g\in G$.
    Then the codes $Q_G^{H_1}(\ah,\bh)$ and $Q_G^{H_2}(\ah,g^{-1}*\bh*g)$ are equivalent.
\end{lemma}
\begin{proof}
    Let us define the map
    \begin{align*}
        \sigma_g : G/H_1 &\rightarrow G/H_2 \\
        xH_1&\mapsto (xg)H_2.
    \end{align*}
    We have the chain of relationships:
    \begin{align*}
        xH_1=yH_1 &\Leftrightarrow x^{-1}y\in H_1\\
        &\Leftrightarrow g^{-1}x^{-1}yg \in g^{-1}H_1g \\
        &\Leftrightarrow (xg)^{-1}(yg)\in H_2\\
        & \Leftrightarrow (xg)H_2=(yg)H_2 \Leftrightarrow \sigma_g(xH_1)=\sigma_g(yH_1)
    \end{align*}
By reading this chain of relationships both ways, we note that $\sigma_g$ is well defined and injective. Since $[G:H_1]=[G:H_2]$ and $G$ is finite, it must also be surjective, showing that $\sigma_g$ is a well-defined bijection for all $g\in G$. Let $L_1$ ($L_2$) be the homomorphism as in \eqref{eq:defL} defined for the subgroup $H_1$ ($H_2$). Then
    \begin{align*}
        (L_2(h)\circ \sigma_g )(xH_1) &= L_2(h)((xg)H_2)\\
        &=(h(xg))H_2\\
        &=((hx)g)H_2\\
        &=\sigma_g((hx)H_1) = (\sigma_g\circ L_1(h))(xH_1).
    \end{align*}
    Hence, for any $h\in G$ we have
    \begin{align*}
        L_2(h) = \sigma_g\circ L_1(h)\circ \sigma_g^{-1}
    \end{align*}
  Let $P_{\sigma_g}$ be the permutation matrix representing the bijection $\sigma_g$. Then we have
    \begin{align}\label{eq:L_aprime_a}
        \bfL_2(\ah) = P_{\sigma_g} \bfL_1(\ah) P_{\sigma_g}^T.
    \end{align}
Below we write $N_i:=N_G(H_i), i=1,2.$    Similarly, let $R_i, i=1,2$ be the homomorphism in \eqref{eq:defR}. Note that for any $h\in N_1$, $g^{-1}hg \in N_2$ since $H_2=g^{-1}H_1g$. Then, 
    \begin{align*}
        (R_2(g^{-1}hg)\circ \sigma_g)(xH_1) &= R_2(g^{-1}hg)((xg)H_2)\\
        &=((xg)g^{-1}hg)H_2\\
        &=((xh)g)H_2\\
        &=\sigma_g((xh)H_1)\\ 
        &= (\sigma_g\circ R_1(h))(xH_1).
    \end{align*}
    Therefore, for any $h\in N_1$, we have
    \begin{align*}
        R_2(g^{-1}hg) = \sigma_g\circ R_1(h) \circ \sigma_g^{-1}
    \end{align*}
Extending this by linearity to the group algebra, for every $\bh\in F(N_1)$, we obtain
    \begin{align}\label{eq:R_bprime_b}
        \bfR_2(g^{-1}*\bh *g) = P_{\sigma_g} \bfR_1(\bh) P_{\sigma_g}^T.
    \end{align}
    Let $H_X$ and $H_Z$ be the parity-check matrices of the code $Q_G^{H_1}(\ah,\bh)$. Then the parity-check matrices of the code $Q_G^{H_2}(\ah,g^{-1}*\bh*g)$ have the form
    \begin{align*}
        H_X^\prime &= [P_{\sigma_g} \bfL_1(\ah) P_{\sigma_g}^T \mid  P_{\sigma_g} \bfR_1(\bh) P_{\sigma_g}^T]\\
        &= P_{\sigma_g}H_X\begin{bmatrix}
           P_{\sigma_g}^T & \bf0\\
           \bf0 & P_{\sigma_g}^T
        \end{bmatrix}
    \end{align*}
    and
    \begin{align*}
        H_Z^\prime &= [-P_{\sigma_g} \bfR_1(\bh)^T P_{\sigma_g}^T \mid  P_{\sigma_g} \bfL_1(\ah)^T P_{\sigma_g}^T]\\
        &= P_{\sigma_g}H_Z\begin{bmatrix}
           P_{\sigma_g}^T & \bf0\\
           \bf0 & P_{\sigma_g}^T
        \end{bmatrix},
    \end{align*}
and they are obtained by applying the same row and column permutations to the parity-check matrices of the code $Q_G^{H_1}(\ah,\bh)$.
\end{proof}
This lemma shows that it suffices to perform the code search only for a single subgroup from each conjugacy class. 

\subsection{Search Results}
To implement the search, we relied on the Small Groups library in GAP \cite{GAP4} to identify and enumerate all groups and their subgroups
up to a certain order, limiting ourselves to non-abelian groups and their non-normal subgroups (the case of normal subgroups was 
previously covered in \cite{Pryadko2BGA}). 

To find 2BGA codes of a given code length $n$, it is sufficient to search over all groups of order $l=n/2$. However, in our generalized construction, the number of algebraic structures that yield quantum codes of length $n$ is much higher. For instance, consider a quantum code of length $n=16$. While there are only $5$ groups of order $l=8$, there are many more options for group and subgroup pairs $(G,H)$ satisfying the index requirement $[G:H]=l$. Therefore, while it is possible to enumerate all low-weight 2BGA codes over groups of a certain order, it is not computationally feasible to search over all group--subgroup pairs $(G,H)$ that yield a quantum code of a given length. 

For this reason, we restricted our code search to a selected subset of $(G,H)$ pairs using a heuristic designed to filter out overly constrained algebraic structures.
Recall that the building blocks of our codes are the left and right action matrices, as defined in Proposition \ref{prop:LRhomomorpshim}. By Proposition \ref{prop:ImageLandR}, there are $|G|/|\operatorname{Core}_G(H)|$ unique left action matrices $\bfL$, and $|N_G(H)|/|H|$ unique right action matrices $\bfR$ for a given pair $(G,H)$. For the case of 2BGA codes (when $H$ is a normal subgroup), these building blocks are regular left and right group action matrices, and there are exactly $l=[G:H]$ of each for a quotient group $G/H$ of order $l$. However, for the general case where $H$ is not normal, the normalizer $N_G(H)$ can be restrictive if the order of the quotient group $N_G(H)/H$ is not large enough. Indeed, for certain group and subgroup pairs, the order of this quotient group is much lower than $l$, and the codes constructed using these matrices yield only trivial parameters with high probability
On the other hand, the number of unique left action matrices can be significantly larger than the index $l=[G:H]$. Consequently, we exclude group-subgroup pairs $(G,H)$ where the quotient order $|N_G(H)|/|H|$ is too small, as such pairs are highly likely to yield only trivial codes.

We performed an exhaustive search over promising $(G,H)$ pairs, i.e., pairs for which the number of unique left action matrices, $|G|/|\operatorname{Core}_G(H)|$, is larger than $[G:H]$, and the number of unique right action matrices, $|N_G(H)|/|H|$, is not too small. More specifically, we focused mostly on $(G,H)$ pairs such that $|N_G(H)|/|H| = [G:H]/2$ or $[G:H]/3$ and $|G|/|\operatorname{Core}_G(H)|=2[G:H]$ or $(3/2[G:H]$. Our goal was to find binary codes $Q_G^H(\hat{a},\hat{b})$ with stabilizer weights $w_{\hat{a}}=w_{\hat{b}}=3$ and $w_{\hat{a}}=w_{\hat{b}}=4$. Equivalent codes were excluded based on Lemmas \ref{lemma:equivalence1} and \ref{lemma:ConjugateSubgroup}. Since $F$ is the binary field, the search process is equivalent to enumerating all possible subsets of the quotient groups $G/\operatorname{Core}_G(H)$ and $N_G(H)/H$ of size $w_{\ah}-1$ and $w_{\bh}-1$, respectively (recall that we can deduct 1 by the argument in Lemma \ref{lemma:equivalence1}). Therefore, for a given pair $(G,H)$, the size of the search space is given by
\begin{align}\label{eq:SizeSearchSpace}
    N_{\mathrm{T}}=\binom{|G|/|\operatorname{Core}_G(H)|-1}{w_{\ah}-1}\binom{|N_G(H)|/|H|-1}{w_{\bh}-1}.
\end{align}

We had to impose further restrictions on the search space because for certain pairs $(G,H)$, particularly 
for $w_{\ah}=w_{\bh}=4$, exhaustive search proves computationally infeasible. In such cases, we employ a randomized search strategy to cull the entire set of possible codes. Instead of examining each of them, we evaluate only a fraction of candidates selected with probability $p$, chosen as follows. First, we choose a count $N_{\text{target}}$ of codes to be examined, which we set to be between $10^4$ to $10^7$ depending on the group size, and then
analyze or discard the code candidates with probability $p=N_{\text{target}}/N_{\mathrm{T}}$.

\begin{table*}[htbp]
\centering
\caption{Parameters of the new qLDPC codes with notable parameters. The underlying group $G$ is identified by its order $\ell$ and GAP index $m$ via \texttt{SmallGroup(}$\ell, m$\texttt{)}. The \textit{structure} is the output of the GAP function \texttt{StructureDescription(G)}. For exact reproducibility, $s$ denotes the integer index of the subgroup $H$ within the list of non-normal subgroups, generated in GAP 4.14.0 via \texttt{Filtered(AllSubgroups(G), H -> not IsNormal(G,H))}
The generating group algebra elements $\ah$ and $\bh$ are reported as integer arrays; these integers correspond to the indices of the respective coset representatives in the lists \texttt{LeftCosets(G, Core(G,H))} and \texttt{LeftCosets(Normalizer(G,H), H)}. All the distance values in this table are exact. }
\label{tab:group_codes}
\begin{tabular}{@{} ccc cc cc ccc @{}}
\toprule
\multicolumn{3}{c}{\textbf{Code Parameters}} & \multicolumn{2}{c}{\textbf{Group Descriptions}} & \multicolumn{2}{c}{\textbf{Group Algebra}} & \multicolumn{3}{c}{\textbf{GAP Identifiers}} \\
\cmidrule(lr){1-3} \cmidrule(lr){4-5} \cmidrule(lr){6-7} \cmidrule(lr){8-10}

\makebox[1cm][c]{$n$} & \makebox[1cm][c]{$k$} & \makebox[1cm][c]{$d$} & 
$G$ Structure & $H$ Structure & $\ah$ & $\bh$ & 
\makebox[1cm][c]{$\ell$} & \makebox[1cm][c]{$m$} & \makebox[1cm][c]{$s$} \\ \midrule

48 & 8 & 6 & $C_3 \times ((C_{16} \rtimes C_2) \rtimes C_4)$ & $C_{16}$ & $[1, 34, 48]$ & $[1, 6, 12]$ & 384 & 512 & 53 \\
96 & 8 & 10 & $C_3 \times ((C_{16} \rtimes C_2) \rtimes C_4)$ & $C_8$ & $[1, 9, 87]$ & $[1, 21, 23]$ & 384 & 512 & 27 \\
224 & 12 & 16 & $C_7 \times ((C_4 \times C_4) \rtimes C_2)$ & $C_2$ & $[1, 81, 186]$ & $[1, 16, 47]$ & 224 & 53 & 1 \\
84 & 16 & 8 & $C_{21} \times (C_3 \rtimes C_4)$ & $C_6$ & $[1, 14, 71, 89]$ & $[1, 6, 8, 19]$ & 252 & 21 & 6 \\
112 & 16 & 10 & $C_7 \times ((C_4 \times C_4) \rtimes C_2)$ & $C_4$ & $[1, 11, 59, 81]$ & $[1, 15, 23, 25]$ & 224 & 53 & 9 \\
128 & 16 & 12 & $(C_8 \rtimes C_2) \rtimes C_8$ & $C_2$ & $[1, 47, 75, 88]$ & $[1, 8, 12, 19]$ & 128 & 10 & 1 \\
168 & 16 & 15 & $C_7 \times ((C_6 \times C_2) \rtimes C_2)$ & $C_2$ & $[1, 72, 106, 109]$ & $[1, 11, 18, 26]$ & 168 & 33 & 1 \\

\bottomrule
\end{tabular}
\end{table*}

Finding the code distance is a computationally intensive procedure, which we handled as follows. We used the GAP package \texttt{QDistRnd} \cite{Pryadko2022QDistRnd}, which provides an upper bound on the minimum distance, as a quick proxy to filter out poor candidates and identify codes with potentially good parameters. If the returned upper bound was high enough, we used the integer programming approach suggested in \cite{landahl2011faulttolerantquantumcomputingcolor} to compute exact minimum distances. Throughout the paper, codes with exactly computed distances are denoted by $[[n,k,d]]$, whereas codes with only a known upper bound on the distance are denoted by $[[n,k,\leq d]]$.

We list some of the codes we found in Tables \ref{tab:group_codes} and \ref{tab:group_codes_additional}. For the codes in Table \ref{tab:group_codes}, we performed a deeper analysis, including an evaluation of their circuit-level performance. 
A code of weight $w=8$ with parameters $[[84,16,8]]$ that appears in this table was previously reported in 
previous works \cite{Pryadko2BGA} and \cite{tripier2026faulttolerantquantumcomputingtrapped}. We do not know whether these codes are equivalent because testing equivalence of codes (classical or quantum) is computationally hard.
The remaining parameter sets in Table \ref{tab:group_codes} do not appear in earlier literature; in particular, the weight-6 codes with parameters $[[48,8,6]]$ and $[[96,8,10]]$ are not among the codes found by exhaustively examining all 2BGA codes of length $\le 100$ with stabilizer weight 6 in \cite{Pryadko2BGA} (a $[[48,8,6]]$ code mentioned there has stabilizer weight 8). This finding supports the claim that the proposed construction method goes beyond the general
2BGA formalism, giving rise to new quantum LDPC codes. To the best of our knowledge, the remaining codes listed in Table \ref{tab:group_codes} do not appear anywhere in the literature. However, 2BGA codes of length over 100 have not been systematically enumerated, and
it is possible that 2BGA codes with the same parameters, or even equivalent codes, might exist.
Additionally,Table \ref{tab:group_codes_additional} lists several other high-rate  
$Q_G^H(\ah,\bh)$ codes with a potential to perform well in simulations.

\section{\label{sec:SC_Circuit} Syndrome Extraction Circuit}

In this section, we introduce a maximally packed syndrome extraction circuit for the code $Q_G^H(\ah,\bh)$. In Ref.~\cite{Bravyi2024}, the authors presented an algorithm yielding a syndrome extraction cycle for bivariate bicycle (BB) codes of weight $6$. Their circuit minimizes the number of idle qubits and achieves a total depth of $8$, including state initializations and measurements. To accomplish this, they interleaved the $X$- and $Z$-type checks such that no single qubit is operated on twice within a single time step, all while ensuring that this interleaving preserves the validity of the syndrome extraction. In this work, we generalize this construction to any $Q_G^H(\ah,\bh)$ code with an arbitrary stabilizer weight 
less than or equal to $w =w_{\ah} + w_{\bh}$.

Let $Q_G^H(\ah,\bh)$ be a code as defined in Construction~\ref{cons:Generalized2BGA} with block length $n=2[G:H]$.  As in the previous section, here we study binary $Q_G^H(\ah,\bh)$ codes.
With this assumption, the matrices $\bfL(\ah)$ and $\bfR(\bh)$ from Construction~\ref{cons:Generalized2BGA} can be decomposed as
\begin{align*}
    \bfL(\ah) = \sum_{g\in \operatorname{supp}(\ah)} \bfL(g), \quad \text{and} \quad  \bfR(\bh) = \sum_{g\in \operatorname{supp}(\bh)} \bfR(g).
\end{align*}

Denote by $\bfL(g)(i)$ (resp., $\bfR(g)(i)$) the index $j$ such that the $(i, j)$-th entry of the permutation matrix $\bfL(g)$ (resp., $\bfR(g)$) is $1$. This notation corresponds to the inverse permutation convention, identifying the column $j$ for a given row $i$.
Recall that the code $Q_G^H(\ah,\bh)$ comprises $n/2$ $X$-type checks, $n/2$ $Z$-type checks, and $n$ data qubits. Let us label the $X$ and $Z$ checks as $X[i]$ and $Z[i]$, respectively, for $i=1,2,\ldots,n/2$. We partition the data qubits into two blocks, labeled $D_L[i]$ and $D_R[i]$ for $i=1,2,\ldots, n/2$, such that each check $X[i]$ is connected to $D_L[\bfL(g)(i)]$ for all $g \in \operatorname{supp}(\ah)$ and to $D_R[\bfR(g)(i)]$ for all $g\in \operatorname{supp}(\bh)$. Because the transpose of a permutation matrix corresponds to the group inverse (i.e., $\bfL(g)^T = \bfL(g^{-1})$ and $-\bfR(g)^T = \bfR(g^{-1})$), it follows that each check $Z[i]$ is connected to $D_L[\bfR(g^{-1})(i)]$ for all $g\in\operatorname{supp}(\bh)$ and to $D_R[\bfL(g^{-1})(i)]$ for all $g\in\operatorname{supp}(\ah)$.

For each edge incident to an $X$-type check in the Tanner graph, the syndrome computation includes performing a CNOT gate controlled by the check qubit and targeting the data qubit at the other end of the edge. Similarly, for each edge incident to a $Z$-type check, we need to 
perform a CNOT gate controlled by the corresponding data qubit and targeting the $Z$ check. For the valid syndrome extraction, these CNOT gates must be carefully scheduled to follow each other.
A straightforward approach is to perform all $X$-type CNOT operations before all $Z$-type CNOT operations, or vice versa. 
For optimized scheduling, we can interleave $X$- and $Z$-type CNOT gates to reduce the overall circuit depth, thereby minimizing the idle time of the qubits.
In Algorithm \ref{alg:schedule_cnot}, we introduce a subroutine for scheduling interleaved $X$-type and $Z$-type CNOT gates. This function pairs each $X$-type CNOT round with a compatible $Z$-type CNOT round, ensuring that no single data qubit is targeted by multiple gates within a single time step. Ultimately, this subroutine will play a role in the algorithm defining the syndrome cycle of a $Q_G^H(\ah,\bh)$ code.

\begin{algorithm}[htbp]
\caption{Subroutine for CNOT Scheduling}
\label{alg:schedule_cnot}
\begin{algorithmic}[1]

\MultilineInput{
    Physical data block $D_p \in \{D_L, D_R\}$. \\
    Permutation matrix type $\bfM_{\ph} \in \{\bfL, \bfR\}$. \\
    Equal-length sequence pairs $(\vec{g}_X, \vec{h}_X)$ and $(\vec{g}_Z, \vec{h}_Z)$ of elements from a group $G_p$. \\
    Sequences $\vec{f}_X,\vec{f}_Z$ containing the same elements from a group $G_u$
       where $G_u \leq G_p$ or $G_p \leq G_u$.
}
\MultilineOutput{
    Updated quantum states of $D_p, D_u, X,$ and $Z$.
}
\Statex
\Function{CnotSchedule}{$D_p, \bfM_{\ph}, \vec{g}_X, \vec{h}_X, \vec{g}_Z, \vec{h}_Z, \vec{f}_X,\vec{f}_Z$}
    \State $D_u \leftarrow \{D_L, D_R\} \setminus \{D_p\}$
    \State $\bfM_{\uh } \leftarrow \{\bfL, \bfR\} \setminus \{\bfM_{\ph}\}$
    \State Let $N$ be the dimension (row/column size) of $\bfM_{\uh }$
    \Statex
    \For{$k = 1$ to $|\vec{g}_X|$} \label{line:1_to_gX_start}
        \State Simultaneously for $i = 1$ to $N$:
        \State \quad $\text{CNOT}(X[i], D_p[\bfM_{\ph}(g_{X,k})(i)])$
        \State \quad $\text{CNOT}(D_u[i], Z[\bfM_{\ph}(g_{Z,k})(i)])$ \label{line:1_to_gX_end}
    \EndFor
    \Statex
    \For{$k = 1$ to $|\vec{f}_X|$}\label{line:1_to_fX_start}
        \State Simultaneously for $i = 1$ to $N$:
        \State \quad $\text{CNOT}(X[i], D_u[\bfM_{\uh }(f_{X,k})(i)])$
        \State \quad $\text{CNOT}(D_p[i], Z[\bfM_{\uh }(f_{Z,k})(i)])$\label{line:1_to_fX_end}
    \EndFor
    \Statex
    \For{$k = 1$ to $|\vec{h}_X|$}\label{line:1_to_hX_start}
        \State Simultaneously for $i = 1$ to $N$:
        \State \quad $\text{CNOT}(X[i], D_p[\bfM_{\ph}(h_{X,k})(i)])$
        \State \quad $\text{CNOT}(D_u[i], Z[\bfM_{\ph}(h_{Z,k})(i)])$\label{line:1_to_hX_end}
    \EndFor
\EndFunction

\end{algorithmic}
\end{algorithm}

In Algorithm \ref{alg:syndrome_extraction}, we describe a syndrome extraction circuit for an arbitrary code $Q_G^H(\ah,\bh)$. 
Before presenting this algorithm, let us introduce the notation and the scheduling strategies used in it. Since the \textsc{CnotSchedule} subroutine is asymmetric with respect to its inputs, we assign one of the group algebra elements as the primary block, denoted $\ph \in \{\ah, \bh\}$, and the other as the auxiliary block, denoted $\uh$. Swapping the assignments of $\ah$ and $\bh$ yields distinct, valid scheduling cycles, which allows the algorithm to cover both possible syndrome extraction schedules. Furthermore, in the case of odd stabilizer weights $w_{\ph}$, we were unable to identify a perfectly symmetric syndrome extraction cycle. To address this, the algorithm introduces a symmetry-breaking step by assigning a primary syndrome measurement basis, $S \in \{X, Z\}$, and a complementary basis $U$.  Alternating the initializations, specific CNOT gates, and measurements of $S$ and $U$ supports the design of a valid CNOT scheduling with minimum circuit depth. An example of an explicit input configuration for the case of $w_{\ah}=w_{\bh}=3$ and $w_{\ah}=w_{\bh}=4$ in  Algorithm \ref{alg:syndrome_extraction} is presented in Section~\ref{sec:CircuitLevelSimulations}.

\begin{algorithm}[htbp]
\caption{Syndrome Extraction Circuit for $d$ Rounds}
\label{alg:syndrome_extraction}
\begin{algorithmic}[1]

\MultilineInput{
A $Q_G^H(\ah,\bh)$ code. \\
An integer $d$ specifying the number of rounds.
}
\MultilineOutput{
$X$ and $Z$ syndromes for $d$ rounds.
}
\Statex
\State Choose an ordering $(\ph,\uh)\in\{(\ah,\bh),(\bh,\ah)\}$.
\State \textbf{If} $\ph=\ah$, then $(D_p, D_u, \bfM_{\ph}, \bfM_{\uh }) \gets (D_L, D_R, \bfL,\bfR)$,\\ 
\textbf{else} $(D_p, D_u, \bfM_{\ph}, \bfM_{\uh }) \gets (D_R, D_L, \bfR,\bfL)$
\State Let $w_{\ph} = |\operatorname{supp}(\ph)|$.
\State $\vec{f}_X,\vec{f}_Z \gets \text{any permutations of } \operatorname{supp}(\uh )$
\State $N\gets [G:H]$
\Statex
\If{$w_{\ph}$ is odd}
    \State Choose an ordering $(S,U)\in\{(X,Z),(Z,X)\}$. 
    \LongState{$\vec{g}_S, \vec{h}_S \gets$ sequences partitioning $\operatorname{supp}(\ph)$ with $|\vec{g}_S| = |\vec{h}_S| + 1$}
    \LongState{$\vec{g}_U, \vec{h}_U \gets$ sequences partitioning  $\operatorname{supp}(\ph)$ with $\vec{h}_U$ containing the same elements as $\vec{g}_S$} 
    \State Let $g_{s_0}$ and $h_{u_0}$ be the first elements of $\vec{g}_S$ and $\vec{h}_U$.
    \State $\vec{g}_S \gets \vec{g}_S \setminus \{g_{s_0}\}$ and $\vec{h}_U \gets \vec{h}_U \setminus \{h_{u_0}\}$
    \LongState{Let $C_X(i) \equiv \text{CNOT}(X[i], D_p[\bfM_{\ph}(e_X)(i)])$ and $C_Z(i) \equiv \text{CNOT}(D_u[i], Z[\bfM_{\ph}(e_Z)(i)])$, where $e_S = g_{s_0}$ and $e_U = h_{u_0}$}
    \State $\forall i \in [N]$: $R_S(S[i])$
    \For{$r = 1$ to $d$}
        \State $\forall i \in [N]$: $R_U(U[i])$, \quad $C_S(i)$\label{line:InitWpOdd}
        \State \Call{CnotSchedule}{$D_p, \bfM_{\ph}, \vec{g}_X, \vec{h}_X, \vec{g}_Z, \vec{h}_Z, \vec{f}_X,\vec{f}_Z$}\label{line:CnotScheduleWpOdd}
        \State $\forall i \in [N]$: $C_U(i)$, \quad $M_S(S[i])$\label{line:MZWpOdd}
        \State $\forall i \in [N]$: $M_U(U[i])$, \quad $R_S(S[i])$\label{line:MXWpOdd}
    \EndFor
\Else
    \LongState{$\vec{g}_X, \vec{h}_X \gets$ any equal-length sequences partitioning $\operatorname{supp}(\ph)$}
    \LongState{$\vec{g}_Z, \vec{h}_Z \gets$ any equal-length sequences partitioning $\operatorname{supp}(\ph)$ with $\vec{h}_Z$ containing the same elements as $\vec{g}_X$}
    \For{$r = 1$ to $d$}
        \State $\forall i \in [N]$: $R_X(X[i])$, \quad $R_Z(Z[i])$\label{line:InitWpEven}
        \State \Call{CnotSchedule}{$D_p, \bfM_{\ph}, \vec{g}_X, \vec{h}_X, \vec{g}_Z, \vec{h}_Z, \vec{f}_X,\vec{f}_Z$}\label{line:CnotScheduleWpEven}
        \State $\forall i \in [N]$: $M_X(X[i])$, \quad $M_Z(Z[i])$\label{line:MeasureWpEven}
    \EndFor
\EndIf
\end{algorithmic}
\end{algorithm}

To verify the correctness of this extraction cycle, one must ensure that the following three conditions are satisfied: (1) every required CNOT gate appears exactly once, (2) no qubit is acted upon twice within the same time step, and (3) the interleaving of $X$- and $Z$-type operations does not alter the final measured syndrome. They are verified in the following proposition whose proof appears in Appendix \ref{appendix:SC_Circuit_Proof}.
\begin{proposition}\label{prop:SyndromeCircuit}
    For any binary $Q_G^H(\ah,\bh)$ code, Algorithm \ref{alg:syndrome_extraction} yields a syndrome extraction circuit of depth $w_{\ah}+w_{\bh}+2$ per cycle, including state initializations and measurements.
\end{proposition}

    When executing $d$ consecutive cycles, in the case that both $w_{\ah}$ and $w_{\bh}$ are odd, our circuit requires one additional time step, yielding a total depth of $d(w_{\ah}+w_{\bh}+2)+1$.

The authors of \cite{strikis2026highperformancesyndromeextractioncircuits} introduced a generalized syndrome extraction protocol for any CSS code, while a packed syndrome extraction circuit specifically for the family of cyclic HGP codes was proposed in \cite{aydin2026cyclichypergraphproductcode}. The circuit designs in these works rely on a non-interleaved approach, where all $Z$-type CNOT gates precede all $X$-type CNOT gates within the same syndrome round. We note that the Tanner graph of the code $Q_G^H(\ah,\bh)$ accommodates both of these non-interleaved approaches, which would yield a total depth of $\min(w_{\ah},w_{\bh})+d(w+2)$ for $d$ rounds. Interleaving the $X$-type and $Z$-type CNOTs and leveraging the specific algebraic structure of the code $Q_G^H(\ah,\bh)$, we can reduce the depth to $d(w+2)+1$ when both $w_{\ah}$ and $w_{\bh}$ are odd and to $d(w+2)$ when at least one of them is even.

We further note that our syndrome extraction circuit is maximally packed in the sense that there are no idle qubits during any CNOT round for the case when at least one of $w_{\ah}$ and $w_{\bh}$ is even. When both $w_{\ah}$ and $w_{\bh}$ are odd, only half of the data qubits are idle, and this occurs only in the first and last CNOT rounds.

We conclude this section with a remark on the number of distinct syndrome extraction cycles that can be generated using Algorithm \ref{alg:syndrome_extraction}. A simple combinatorial argument, detailed in Appendix \ref{appendix:Counting_Proof},
yields the number of distinct sequence configurations in the form
\begin{align}\label{eq:NoSC_Circuit}
    (w_{\ah}! w_{\bh}!)^2 \left(\frac{\nu(w_{\ah})}{\binom{w_{\ah}}{\lfloor w_{\ah}/2 \rfloor}} + \frac{\nu(w_{\bh})}{\binom{w_{\bh}}{\lfloor w_{\bh}/2 \rfloor}}\right),
\end{align}
where $\nu(w)=1$ if $w$ is even and $\nu(w)=2$ if $w$ is odd. This count gives an upper bound on the number of unique physical circuits
since different sequences may yield identical circuits due to code symmetries. Note that even though our algorithm provides an explicit description of maximally packed syndrome extraction cycles for a broad class of LDPC codes with any stabilizer weight, it does not necessarily cover all possible valid syndrome cycles. However, it generates a large family of scheduling options. For example, for a weight-$6$ code with $w_{\ah}=w_{\bh}=3$, it generates 1728 distinct configurations. All of these configurations apply to a broad class of LDPC codes, including all weight-$6$ BB codes of \cite{Bravyi2024}. For any weight-$8$ code with $w_{\ah}=w_{\bh}=4$, Algorithm \ref{alg:syndrome_extraction} yields 110,592 distinct configurations. For a given code, it is possible to iterate over these different configurations to minimize hook errors. 

\section{\label{sec:CircuitLevelSimulations} Circuit-Level Simulations}
We evaluated the performance of selected $Q_G^H(\ah,\bh)$ codes under a standard circuit-level noise model. In this model, the noise is parameterized by a single physical error rate $p$ and is applied according to the following rules:
\begin{itemize}
    \item \textbf{State preparation:} Qubits initialized in the $X$- or $Z$-basis experience a $Z$- or $X$-type error, respectively, with probability $p$ immediately following initialization.
    \item \textbf{Gate operations:} Every CNOT gate is followed by a two-qubit depolarizing channel, which applies an error with probability $p$
(choosing one of the 15 non-trivial two-qubit Pauli errors uniformly at random) and acts as identity with probability $1-p$.
    \item \textbf{Idle qubits:} Any qubit remaining idle during a given time step is acted upon by a single-qubit depolarizing channel with probability $p$,
    which acts by a single-qubit Pauli, choosing $X$, $Y$, or $Z$ uniformly at random.
    \item \textbf{Measurements:} Ancilla qubits measured in the $X$- or $Z$-basis experience a $Z$- or $X$-type error, respectively, with probability $p$ immediately prior to the measurement.
\end{itemize}

For a $Q_G^H(\ah,\bh)$ code with distance $d$, we ran Algorithm \ref{alg:syndrome_extraction} for $d$ rounds. For group algebra elements of the form $\ah=a_1+a_2+\ldots+a_{w_{\ah}}$ and $\bh= b_1+b_2+\ldots+b_{w_{\bh}}$, we used the configuration $(\ph, \uh)=(\ah, \bh)$, $\vec{h}_Z=\vec{g}_X = (a_2,a_3)$, $\vec{h}_X=\vec{g}_Z=(a_1,a_4)$, $\vec{f}_X=(b_1,b_2,b_3,b_4)$, and $\vec{f}_Z=(b_2,b_1,b_4,b_3)$ for our weight-$8$ code simulations. Similarly, for our weight-$6$ codes with $w_{\ah}=w_{\bh}=3$, we used the configuration $(\ph,\uh )=(\ah,\bh)$, $(S,U)=(Z,X)$, $\vec{g}_Z=(a_1,a_3)$, $\vec{h}_Z=(a_2)$, $\vec{g}_X=(a_2)$, $\vec{h}_X=(a_1,a_3)$, $\vec{f}_X=(b_2,b_1,b_3)$, and $\vec{f}_Z=(b_1,b_2,b_3)$. While we fixed the orderings for our simulations across all codes, one can iterate over the extensive set of configurations provided by Algorithm \ref{alg:syndrome_extraction} to achieve even greater circuit-level error suppression. 

We performed our circuit-level simulations using \texttt{Stim} \cite{Gidney2021stimfaststabilizer} and \texttt{qLDPC} \cite{perlin2023qldpc} packages. 
Following the current standard for benchmarking qLDPC codes in the literature,
we used the Belief Propagation with Ordered Statistics Decoding (BP-OSD) algorithm \cite{Panteleev2021degeneratequantum, PhysRevResearch.2.043423}. Specifically, we utilized the \texttt{stimbposd} \cite{Higgott2022stimbposd} implementation with $10,000$ BP iterations and a combination sweep depth of $10$. 

We numerically computed the logical error rates for $Z$- and $X$-type observables separately. Let $N_{e,d,W}(p)$ be the number of $W$-type logical errors observed in $N_{s,d,W}(p)$ Monte Carlo samples after $d$ rounds of syndrome extraction at a physical error probability $p$, where $W\in\{X,Z\}$. We define the logical error rate per cycle for each basis as follows: 
\begin{align*}
    p_{L,W}(p) = 1 -\left(1-\frac{N_{e,d,W}(p)}{N_{s,d,W}(p)}\right)^{1/d}, \quad W\in\{X,Z\},
\end{align*}
which is the complement to the success probability per syndrome cycle. Throughout the paper, we report the logical error rate per cycle, given by
\begin{align*}
    p_L(p)=1-(1-p_{L,X}(p))(1-p_{L,Z}(p)),
\end{align*}
representing the complement of the probability of success in both bases simultaneously. The pseudo-threshold reported for each code is defined as the solution of the break-even condition:
\begin{align*}
    p_L(p_{th})=1-(1-p_{th})^{k}.
\end{align*}
Here $1-(1-p)^{k}$ represents the probability of at least one error occurring in $k$ qubits. In Figures \ref{fig:circuit_level_weight_6} and \ref{fig:circuit_level_weight_8}, we plot the logical error rates versus physical error rates of the codes listed in Table \ref{tab:group_codes}, decoded using the BP-OSD decoder under the standard circuit-level noise model. The curves in these figures are obtained by fitting the numerical data points in the sub-threshold regime to the heuristic fitting formula $\Tilde{p}_L(p)= p^{d_{\mathrm{circ}}/2}e^{\alpha+\beta p + \gamma p^2}$. The parameter $d_{\mathrm{circ}}$, which appears in the exponent of $p$, represents circuit-level distance of the code. This is defined as the minimum number of faulty operations required to cause an undetectable error during a syndrome extraction cycle. An upper bound on this parameter is obtained by solving the optimization problem introduced in \cite[Supplemental Sec.~2]{Bravyi2024}. For each code, we set the number of trials to at least $50{,}000$ to obtain these upper bounds. The fitting parameters, as well as the upper bounds on the circuit-level distance, are provided in the Appendix \ref{appendix:FitParams}.

\begin{figure}[htbp]
    \centering
    % Subfigure (a) - Weight 6
    \begin{subfigure}[b]{0.9\columnwidth}
        \centering
        \includegraphics[width=\textwidth]{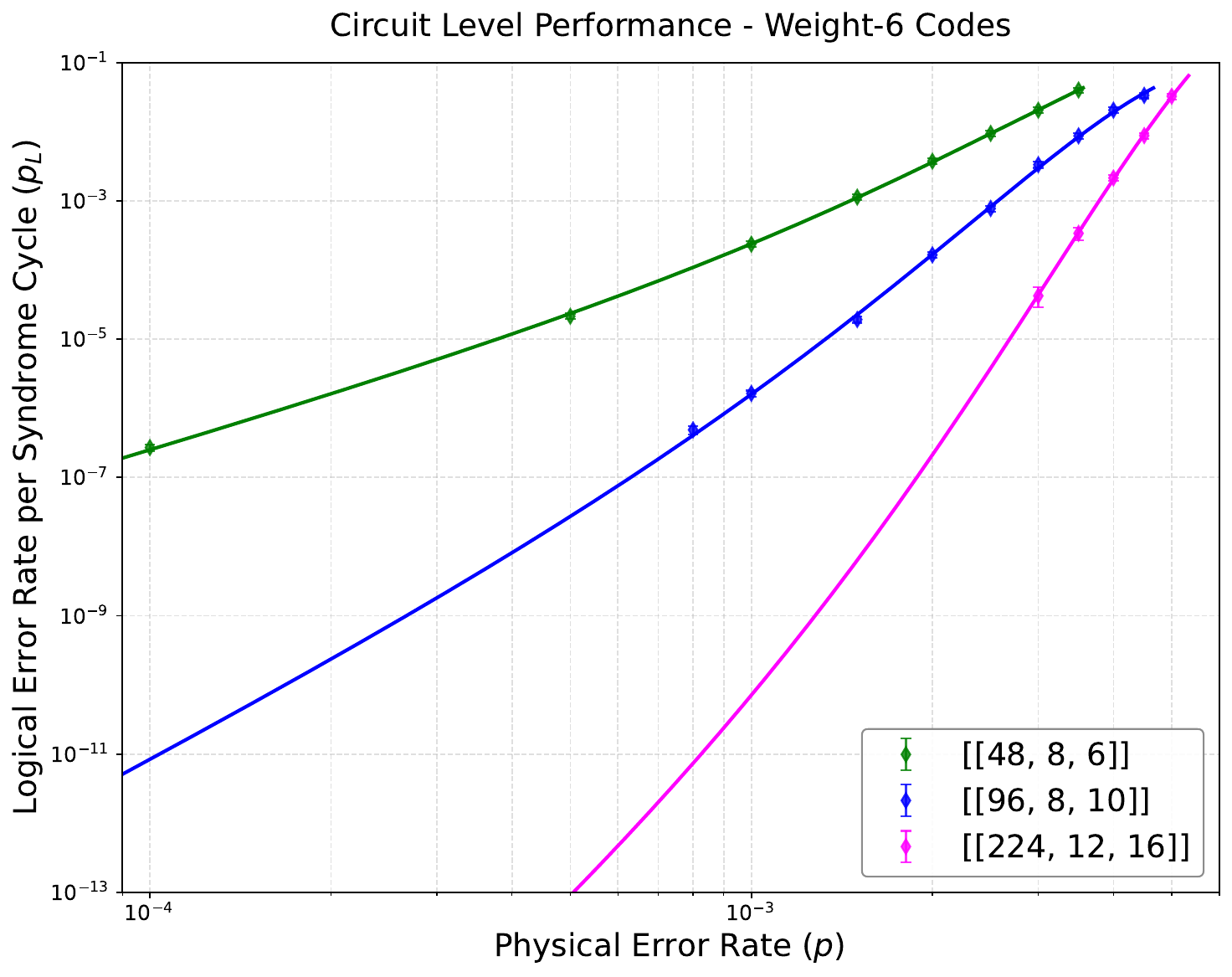}
        \caption{Weight-6 codes}
        \label{fig:circuit_level_weight_6}
    \end{subfigure}
    
    \vspace{2em} % Adds vertical space between the subfigures

    % Subfigure (b) - Weight 8
    \begin{subfigure}[b]{0.9\columnwidth}
        \centering
        \includegraphics[width=\textwidth]{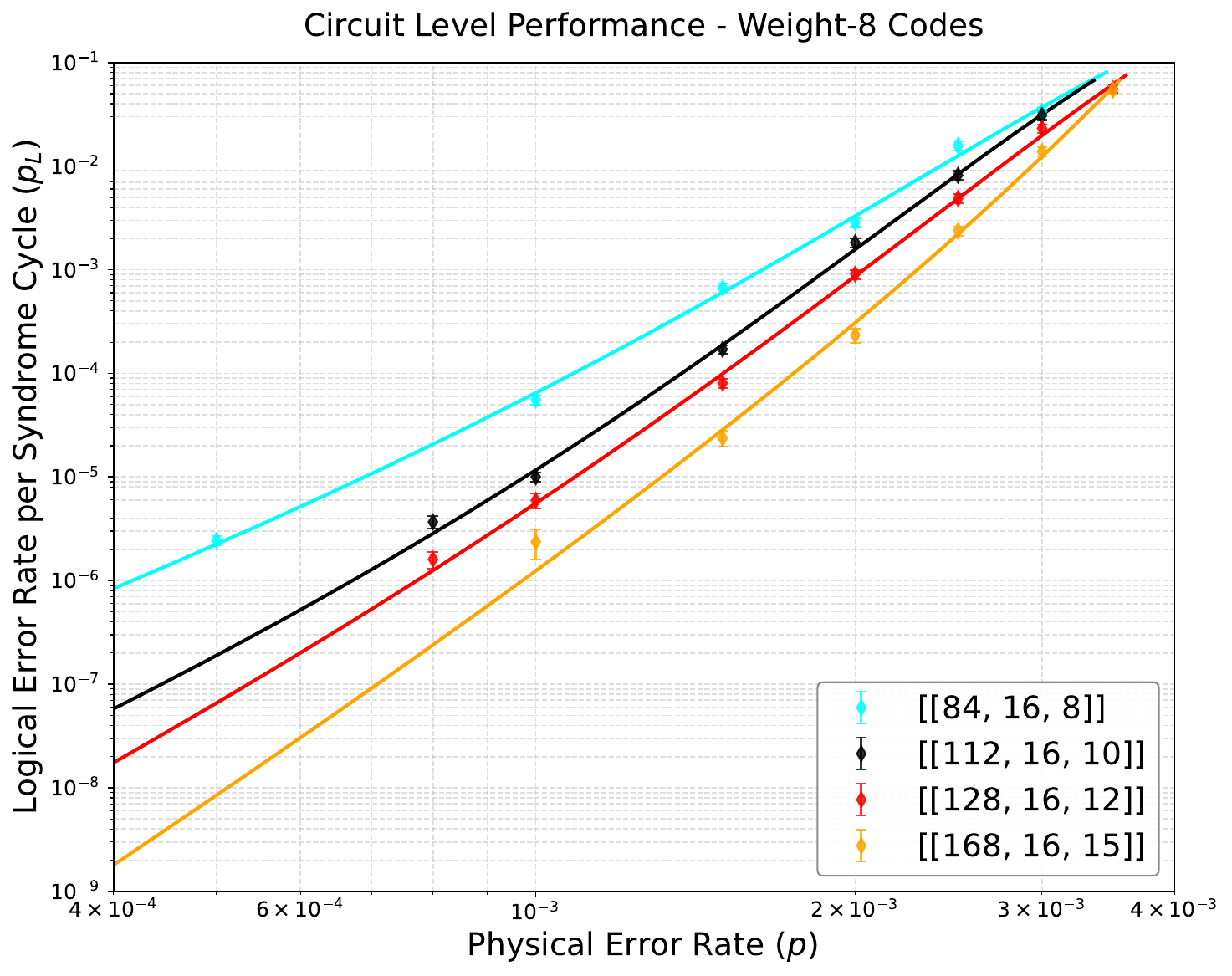}
        \caption{Weight-8 codes}
        \label{fig:circuit_level_weight_8}
    \end{subfigure}
    
    % Main Caption
    \caption{Logical error rate $p_L$ as a function of the physical error rate $p$ for (a) weight-6 and (b) weight-8 $Q_G^H(\ah,\bh)$ codes under the standard circuit-level noise model. The numerical data points represent Monte Carlo simulations, while the solid curves are obtained using the fitting formula $\Tilde{p}_L(p)= p^{d_{\mathrm{circ}}/2}e^{\alpha+\beta p + \gamma p^2}$. All codes were decoded using the BP-OSD decoder.}
    \label{fig:combined_circuit_level}
\end{figure}

\begin{table*}[htbp]
\centering
\caption{Summary of simulation results for selected $Q_G^H(\ah,\bh)$ codes under the standard circuit-level noise model. All logical error rates ($p_L$) and pseudo-thresholds are reported for the BP-OSD decoder. The net encoding rate is defined as $1/\lceil 2n/k \rceil$ to account for both data and ancilla qubits.}
\label{tab:summarylogicalerrorrates}
\renewcommand{\arraystretch}{1.2} 
\begin{tabular}{@{} c | c | c | c | c | c | c @{}}
\hline
\textbf{Code} & \textbf{Weight} & \begin{tabular}{@{}c@{}}\textbf{Net Encoding} \\ \textbf{Rate}\end{tabular} & \textbf{Pseudo-threshold} & $p_L(10^{-3})$ & ${p_L(10^{-4})}$ & ${k d^2/n}$ \\
\hline
$[[48, 8, 6]]$      & 6 & $1/12$ & $0.36\%$ & $2\times10^{-4}$  & $3\times10^{-7}$  & 6.00 \\
$[[96, 8, 10]]$     & 6 & $1/24$ & $0.47\%$ & $2\times10^{-6}$  & $8\times10^{-12}$ & 8.33 \\
$[[224, 12, 16]]$   & 6 & $1/38$ & $0.54\%$ & $7\times10^{-11}$ & $2\times10^{-19}$ & 13.71 \\
\hline
$[[84, 16, 8]]$     & 8 & $1/11$ & $0.35\%$ & $6\times10^{-5}$  & $3\times10^{-9}$  & 12.19 \\
$[[112, 16, 10]]$   & 8 & $1/14$ & $0.34\%$ & $1\times10^{-5}$  & $9\times10^{-11}$ & 14.29 \\
$[[128, 16, 12]]$   & 8 & $1/16$ & $0.36\%$ & $6\times10^{-6}$  & $9\times10^{-12}$ & 18.00 \\
$[[168, 16, 15]]$   & 8 & $1/21$ & $0.36\%$ & $2\times10^{-6}$  & $2\times10^{-13}$ & 21.43 \\
\hline
\end{tabular}
\end{table*}

The weight-$6$ $Q_G^H(\ah,\bh)$ codes introduced in this work are competitive with BB codes. We estimated the threshold for the weight-6 code family to be $0.0065$. For the weight-$8$ codes, the threshold decreases to $0.0035$. However, the weight-$8$ codes allow for higher rates and larger distances. Consequently, their logical error rates remain competitive, particularly at lower physical error rates, as the higher distances result in a steeper suppression of errors in the sub-threshold regime. 

\subsection{Performance under various decoders}

We also evaluated a subset of codes using recently introduced decoders such as the beam search decoder \cite{ye2025beamsearchdecoderquantum} and the relay-BP decoder \cite{muller2025improvedbeliefpropagationsufficient}. In our simulations, we used a beam width of $32$ for the beam search decoder. As noted in Ref. \cite{ye2025beamsearchdecoderquantum}, this decoder is a promising candidate for trapped-ion architectures due to its fast 99.9th percentile execution time. We note that the configuration we used is designed for real-time fast decoding, and better error suppression is possible with a more aggressive configuration for this decoder. For the Relay-BP decoder, we utilized the configuration referred to as 'Relay-BP-5' in Ref. \cite{muller2025improvedbeliefpropagationsufficient}. We also chose the memory strengths matching those used for Gross code in the same paper. We note that further optimization of these parameters could yield improvements in both accuracy and decoding time. 

In Figure \ref{fig:decoder_performance}, we plot the logical error rate as a function of physical error rate for the decoders discussed above. Specifically, we present results for the weight-$6$ codes $[[48,8,6]]$, $[[96,8,10]]$ and the weight-$8$ code $[[128,16,12]]$. 

\begin{figure*}[t]
    \centering
    % --- Subfigure (a) ---
    \begin{subfigure}[b]{0.32\textwidth}
        \centering
        \includegraphics[width=\linewidth]{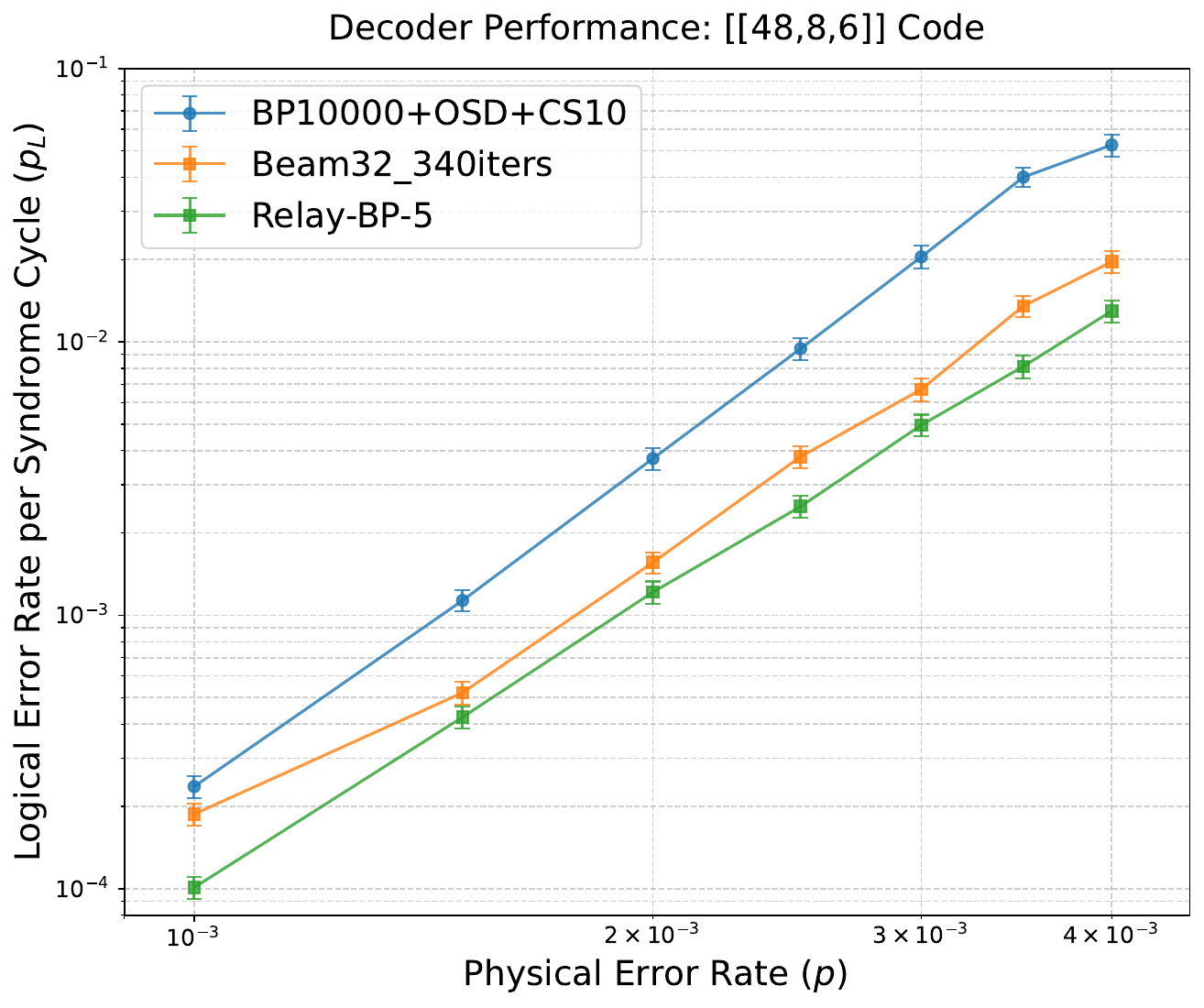}
        \caption{}
        \label{fig:weight6}
    \end{subfigure}
    \hfill
    % --- Subfigure (b) ---
    \begin{subfigure}[b]{0.32\textwidth}
        \centering
        \includegraphics[width=\linewidth]{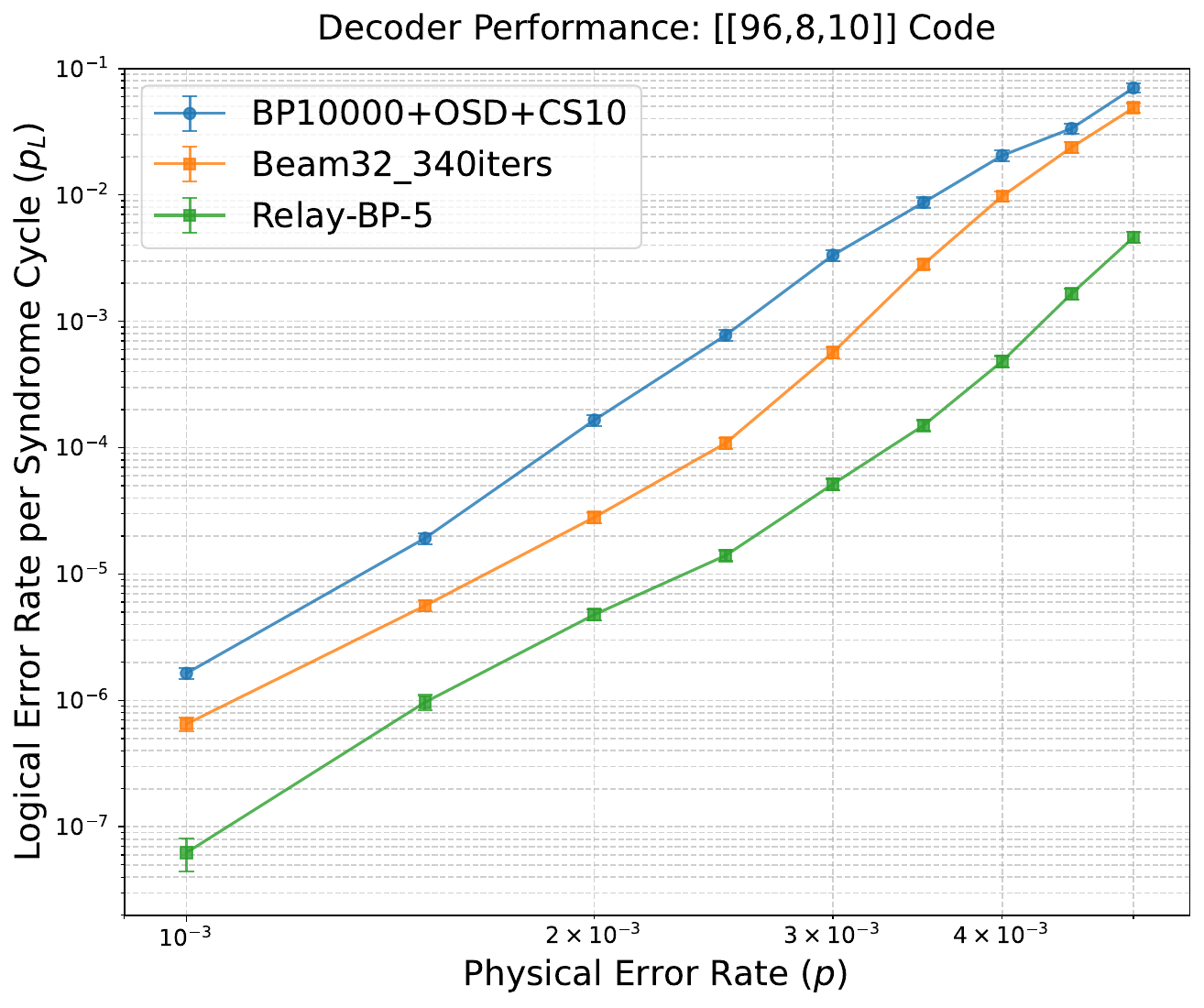}
        \caption{}
        \label{fig:weight8}
    \end{subfigure}
    \hfill
    % --- Subfigure (c) ---
    \begin{subfigure}[b]{0.32\textwidth}
        \centering
        \includegraphics[width=\linewidth]{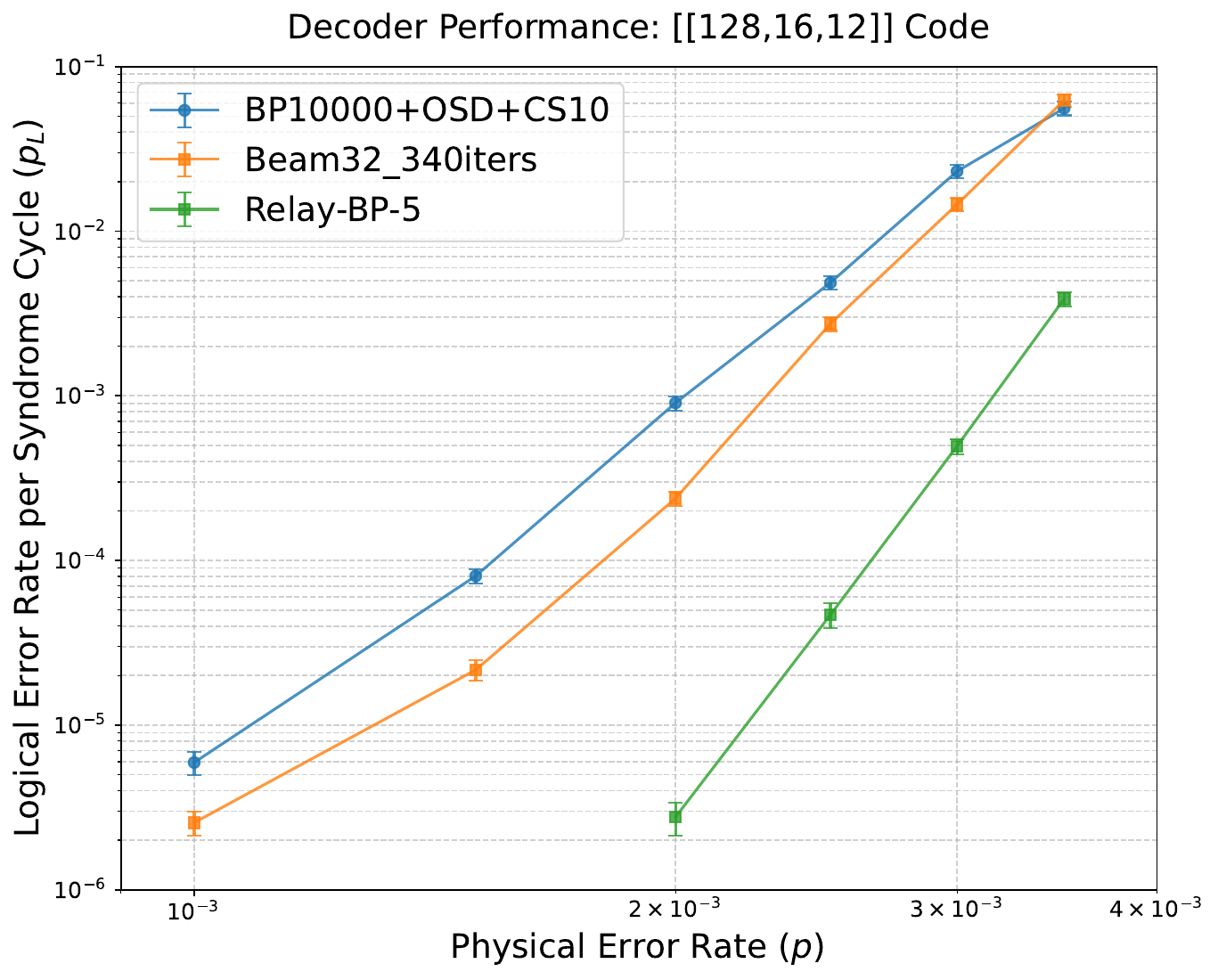}
        \caption{}
        \label{fig:realtime}
    \end{subfigure}

    \caption{Logical error rate $p_L$ as a function of physical error rate $p$ for various $Q_G^H(\ah, \bh)$ codes under different decoders: (a) the $[[48,8,6]]$ code, (b) the $[[96,8,10]]$ code, and (c) the $[[128,16,12]]$ code.}
    \label{fig:decoder_performance}
\end{figure*}

As can be seen from these figures, our codes performed better under both Relay-BP and beam search decoder than under BP-OSD. For the $[[48,8,6]]$ code, the pseudo-threshold increases from $0.0036$ to $0.0046$ for the beam search decoder and to $0.0054$ for the Relay-BP decoder. For the $[[96,8,10]]$ code, the pseudo-threshold increases from $0.0047$ to $0.0063$ for Relay-BP decoder. Similarly, for the $[[128,16,12]]$ code, the pseudo-threshold increases from $0.0036$ to $0.0045$ using the Relay-BP decoder. At a physical error rate of $2\times 10^{-3}$, the beam search decoder achieves a $2.4\times$ reduction, and the relay-BP decoder achieves a $3.1\times$ reduction in the logical error rate compared to BP-OSD for the $[[48,8,6]]$ code. For the $[[96,8,10]]$ code, the beam search and relay-BP decoders yield $5.88\times$ and $34.63\times$ reductions, respectively. We observe a particularly significant improvement with the relay-BP decoder over BP-OSD for the weight-8 $[[128,16,12]]$ code: while the beam search decoder provides a $3.8\times$ reduction, the relay-BP decoder achieves $327.6\times$ reduction in the logical error rate. Note that while the relay-BP decoder shows strong error suppression capabilities, especially for the weight-8 $[[128,16,12]]$ code, the beam search decoder (even in the fast, non-aggressive configuration used here) outperforms the BP-OSD baseline while maintaining the fast 99.9th percentile execution time, which is critical for real-time decoding architectures.

These observations strongly suggest that our codes also perform well under alternative decoding strategies. The real-time decodability of our codes, even with their requirement for all-to-all connectivity, makes them attractive candidates for architectures supporting long-range interactions.

\section{Sequences of Codes from Covering Graphs}\label{sec:CoverGraphs}
In this section, we study sequences of codes constructed from covering graphs of 2BGA codes, a construction that is conceptually distinct from the coset-based methods discussed throughout the previous sections. In classical coding theory, lifting a small base graph to a larger covering graph is a standard and powerful technique for constructing high-performing families of LDPC codes. In practice, this derived graph is commonly constructed by applying `copy-and-permute' operations to a small structural blueprint known as a protograph \cite{2003IPNPR.154C...1T}. Beyond protographs, other graph-covering techniques used for constructing classical codes include unwrapping block codes into LDPC convolutional codes \cite{5695133}, using topological voltage graphs to construct block codes \cite{4595095}, and optimizing voltage assignments to maximize the girth of tailbiting codes \cite{6084747}. An important subclass of graph-cover constructions arises when the lifting voltages are assigned from a cyclic group. 
This approach yields quasi-cyclic LDPC (QC-LDPC) codes, which algebraically map the lifting operations to shifts of circulant matrices \cite{1317123, 1362891}. Since the covering graph is locally isomorphic to the base graph, a sequence of codes constructed from a well-designed base code is expected to perform well under belief propagation (BP) since BP structurally is a local message-passing algorithm.

Extending these graph-cover techniques to quantum codes is significantly more challenging due to the orthogonality constraint required by the CSS construction. However, by utilizing constructions based on the product of two classical codes such as hypergraph product (HGP) codes \cite{tripier2026faulttolerantquantumcomputingtrapped} or lifted product (LP) \cite{10.1145/3519935.3520017} codes, these classical lifting techniques can be extended
to the quantum setting. Moving beyond product-based constructions, a more general framework for lifting arbitrary quantum CSS codes was recently developed in Ref. \cite{10931129}. More recently, by studying voltage assignments on the Tanner graphs of BB codes, 
Symons, Rajput, and Browne \cite{symons2025sequencesbivariatebicyclecodes} established conditions guaranteeing that the code described by the covering Tanner graph is also a BB code. 
%In this section, we formulate their approach in an algebraic language, generalizing and simplifying their results.

Building on these concepts, in this section, we formally define graph-cover-based sequences of 2BGA codes, translating the topological structure into a purely group-theoretic language. 

%\zac{Since $G$ is already used throughout the paper for a group, using $G=(V,E)$ for a graph here is confusing. I suggest denoting a general graph by $\Gamma=(V,E)$ or $\tau=(V,E)$ in this paragraph. This is especially important in this section, where both groups and graphs appear simultaneously.}
We begin by recalling the concept of a covering graph. Given a graph $\tau=(V,E)$, let $$\cN(v):=\{w\in V: vw\in E\}$$ be the neighborhood of a vertex $v\in V$.

\begin{definition}[Covering Graph]\label{def:CoverGraph}
    Let $\tau=(V,E)$ and $\tilde{\tau}=(\tilde{V},\tilde{E})$ be two graphs, and let $f:\tilde{V}\to V$ be a surjection.  If for every vertex $\tilde{v} \in \tilde{V}$,  $f$ maps bijectively ${\cN(\tilde{v})}$ onto $\cN(f(\tilde{v}))$, then $\tilde{\tau}$ is called a covering graph of the base graph $\tau$. Furthermore, if $|f^{-1}(v)|=h$ for every $v\in V$, then $\tilde{\tau}$ is called an $h$-fold cover of $\tau$.  
\end{definition}
Intuitively, an $h$-fold cover can be understood as replacing each vertex of the base graph with a set of $h$ vertices, while preserving the local neighborhood structure.
Our main result in this section is the construction of sequences of 2BGA codes from covering graphs. 

\begin{theorem}\label{theorem:coverCodes}
    Let $Q_G(\ah,\bh)\triangleq Q_G^{\{e\}}(\ah, \bh)$ be a binary 2BGA code of length $n$. Let $(G_i)_{i\ge1}$ be a sequence of groups containing normal subgroups $H_i \trianglelefteq G_i$ such that for all $i\ge 1$, there exist group isomorphisms $\phi_i: G_i/H_i \to G$, which we extend to algebra isomorphisms $\phi_i: \mathbb{F}_2[G_i/H_i] \to \mathbb{F}_2[G]$ by linearity . Let us define the canonical projection homomorphism
    \begin{align*}
        \pi_i : \mathbb{F}_2[G_i] &\to \mathbb{F}_2[G_i/H_i]\\
        \sum_{g\in G_i}\alpha_g g &\mapsto \sum_{g\in G_i}\alpha_g(gH_i).
    \end{align*}
    If the group-algebra elements $\ah_i,\bh_i\in \mathbb{F}_2[G_i]$ are chosen such that $|\operatorname{supp}(\ah_i)|=|\operatorname{supp}(\ah)|$ and $|\operatorname{supp}(\bh_i)|=|\operatorname{supp}(\bh)|$ and they satisfy 
    
    \begin{align*}
        \phi_i\big(\pi_i(\ah_i)\big) = \ah \quad \text{and} \quad \phi_i\big(\pi_i(\bh_i)\big) = \bh,
    \end{align*}
    then the Tanner graph of the code $Q_{G_i}(\ah_i,\bh_i)$ is an $h_i$-fold covering graph of the Tanner graph of the code $Q_{G}(\ah,\bh)$, where $h_i = |H_i|$. 
\end{theorem}
We say that the code $Q_{G_i}(\ah_i,\bh_i)$ of length $n_{h_i}=h_in$ is  an $h_i$-fold cover of the base code $Q_{G}(\ah,\bh)$. Before proceeding to the proof, we clarify that the conditions $|\operatorname{supp}(\ah_i)|=|\operatorname{supp}(\ah)|$ and $|\operatorname{supp}(\bh_i)|=|\operatorname{supp}(\bh)|$ in Theorem \ref{theorem:coverCodes} ensure that each element in $\operatorname{supp}(\ah)$ has exactly one chosen lift in $\operatorname{supp}(\ah_i)$, and similarly for $\bh$.

\begin{proof} (of Theorem \ref{theorem:coverCodes})
Let $\tau=\tau_X\cup \tau_Z$ be the Tanner graph of the binary code $Q_G(\ah,\bh)$, where $\tau_X$ and $\tau_Z$ are the bipartite graphs defined by the parity-check matrices $H_X=[\bfL(\ah) \mid \bfR(\bh)]$ and $H_Z=[\bfR(\bh)^T \mid \bfL(\ah)^T]$, respectively. For 2BGA codes, the blocks $\bfL(\cdot)$ and $\bfR(\cdot)$ are linear combinations of the permutation matrices representing the left and right regular actions of $G$ on itself. Hence, the rows and columns of the parity-check matrices can be naturally labeled by the elements of the group $G$.

Consider first the graph $\tau_X=(V_X\cup V_{D_L}\cup V_{D_R}, E_X)$. Here, $V_X=\{X(h) \mid h \in G\}$ is the set of $X$-checks corresponding to the rows in $H_X$, and $V_{D_k}=\{D_k(h) \mid h \in G\}$ for $k \in \{L,R\}$ represent the sets of data qubits corresponding to the first and second halves of the columns in $H_X$. For every check $X(h)$ where $h\in G$, there is an edge connecting $X(h)$ to the data qubit $D_L(h_a^{-1} h)$ for each $h_a \in \operatorname{supp}(\ah)$, and another edge connecting $X(h)$ to the data qubit $D_R(h h_b^{-1})$ for each $h_b \in \operatorname{supp}(\bh)$. 

Similarly, we define the Tanner graph $\tilde{\tau} = \tilde{\tau}_X \cup \tilde{\tau}_Z$ for the code $Q_{G_i}(\ah_i, \bh_i)$. The $X$-Tanner graph $\tilde{\tau}_X = (\tilde{V}_X \cup \tilde{V}_{D_L} \cup \tilde{V}_{D_R}, \tilde{E}_X)$ is defined such that the vertex sets are $\tilde{V}_X = \{\tilde{X}(g) \mid g \in G_i\}$ and $\tilde{V}_{D_k} = \{\tilde{D}_k(g) \mid g \in G_i\}$ for $k\in\{L,R\}$. 

Define a surjective mapping
\begin{align*}
    f : \tilde{V}_X \cup \tilde{V}_Z \cup \tilde{V}_{D_L} \cup \tilde{V}_{D_R} &\to V_X \cup V_Z \cup V_{D_L} \cup V_{D_R} \\
    \tilde{W}(g) &\mapsto W(\phi_i(\pi_i(g))),
\end{align*}
where $W \in \{ X, Z, D_L, D_R \}$. Note that
\begin{align*}
    \cN(\tilde{X}(g)) &= \{ \tilde{D}_L(g_a^{-1} g) \mid g_a \in \operatorname{supp}(\ah_i) \} \\
    &\hspace*{.3in} \cup \{ \tilde{D}_R(g g_b^{-1}) \mid g_b \in \operatorname{supp}(\bh_i) \}, \quad g \in G_i.
\end{align*}

Let $h(g) = \phi_i(\pi_i(g))$ for all $g \in G_i$. By the conditions of the code construction, for any $g_a \in \operatorname{supp}(\ah_i)$ and $g_b \in \operatorname{supp}(\bh_i)$, we have $h(g_a) \in \operatorname{supp}(\ah)$ and $h(g_b) \in \operatorname{supp}(\bh)$. For the vertex $f(\tilde{X}(g)) = X(h(g))$, the neighborhood in the base graph is given as
\begin{align*}
    \cN(X(h(g))) &= \{ D_L(h(g_a)^{-1} h(g)) \mid h(g_a) \in \operatorname{supp}(\ah) \} \\
    &\hspace*{.3in} \cup \{ D_R(h(g) h(g_b)^{-1}) \mid h(g_b) \in \operatorname{supp}(\bh) \}.
\end{align*}

Let $f|_{\cN(\tilde{X}(g))}$ be the restriction of $f$ to $\cN(\tilde{X}(g))$. By definition, we have $f(\tilde{D}_L(g_a^{-1} g)) = D_L(h(g_a^{-1}g))=D_L(h(g_a^{-1})h(g))=D_L(h(g_a)^{-1}h(g))$ and $f(\tilde{D}_R(g g_b^{-1})) = D_R(h(gg_b^{-1})) = D_R(h(g)h(g_b^{-1})) = D_R(h(g)h(g_b)^{-1}) $. Because the code construction requires that $|\operatorname{supp}(\ah_i)| = |\operatorname{supp}(\ah)|$ and $|\operatorname{supp}(\bh_i)| = |\operatorname{supp}(\bh)|$, the mappings $g_a \mapsto h(g_a)$ and $g_b \mapsto h(g_b)$ are bijections between their respective supports. Consequently, $f|_{\cN(\tilde{X}(g))}$ maps the distinct elements of $\cN(\tilde{X}(g))$ one-to-one and onto the distinct elements of $\cN(f(\tilde{X}(g)))$ for all $g\in G_i$. An identical argument establishes that this local bijection also holds for the neighborhoods of the data vertices $\tilde{D}_L(g)$ and $\tilde{D}_R(g)$.
Finally, since $|f^{-1}(X(h^\prime))| = |H_i| = h_i$ for all $h^\prime \in G$, the graph $\tilde{\tau}_X$ is the $h_i$-fold cover of the graph $\tau_X$.

The claims for the $Z$-Tanner graph are proved by following a similar sequence of steps.
\end{proof} 

We remark that Theorem \ref{theorem:coverCodes} recovers the combined results of Theorems 3.1, 3.3, and 3.4 in \cite{symons2025sequencesbivariatebicyclecodes} as a particular case. Namely, their results correspond to the case of the base group to $G=\mathbb{Z}_l \times \mathbb{Z}_m$, the covering group $\tilde{G}=\mathbb{Z}_{ul} \times \mathbb{Z}_{tm}$, and the normal subgroup $H \cong \mathbb{Z}_u \times \mathbb{Z}_t$. In contrast, our theorem applies to any group extension of $G$ by $H_i$, including non-abelian ones. Thus, setting the covering
argument in a group-theoretic framework, we describe a substantially broader class of LDPC codes compared to the earlier work. 

The following example illustrates this algebraic lifting construction in action.
\begin{example}\label{example:coverCodes}
    Let $V_4=\{1,x,y,xy\}$ be the Klein four-group, and let $Q_{V_4}(\ah_{\mathrm{base}},\bh_{\mathrm{base}})$ be a 2BGA base code defined by $\ah_{\mathrm{base}}=1+x$ and $\bh_{\mathrm{base}}=1+y$. Consider the dihedral group of order $8$ given by the presentation $D_4=\langle r,s \mid r^4 = s^2 = (sr)^2 = 1\rangle$, and a normal subgroup $H=\{1,r^2\}$. The quotient group $D_4/H$ consists of the cosets $H=\{1,r^2\}$, $rH=\{r,r^3\}$, $sH=\{s,sr^2\}$, and $srH=\{sr,sr^3\}$. 
    
 Define a group isomorphism $\phi : D_4/H \xrightarrow{\sim} V_4$ via the mappings:
    \begin{align*}
        H &\mapsto 1, & sH &\mapsto y, \\
        rH &\mapsto x, & srH &\mapsto xy.
    \end{align*}
    To construct a valid lift $\ah_{\text{cover}}$ such that $\phi(\pi(\ah_{\text{cover}}))=\ah_{\mathrm{base}}$ while preserving the support size, we must select one element from $H$ and one element from $rH$. This yields four valid choices for the first group algebra element $\ah_{\text{cover}}$:
    \begin{align*}
        \ah_{\text{cover}} \in \{1+r, \ 1+r^3, \ r^2+r, \ r^2+r^3\}.
    \end{align*}
    Similarly, lifting $\bh_{\mathrm{base}}=1+y$ requires choosing one element from $H$ and one from $sH$, yielding four valid choices for the second group algebra element $\bh_{\text{cover}}$:
    \begin{align*}
        \bh_{\text{cover}} \in \{1+s, \ 1+sr^2, \ r^2+s, \ r^2+sr^2\}.
    \end{align*}
    Selecting any pair $(\ah_{\text{cover}}, \bh_{\text{cover}})$ from these two sets satisfies the conditions of Theorem \ref{theorem:coverCodes}. Consequently, this specific assignment yields $16$ valid pairs of $(\ah_{\text{cover}}, \bh_{\text{cover}})$ that form $2$-fold covers of the base code $Q_{V_4}(1+x,1+y)$. However, as will be formalized by the equivalence principles in Proposition \ref{prop:GeneralCoverEquivalence}, the number of strictly unique codes up to isomorphism is considerably less than $16$. $\triangleleft$
\end{example}

\subsection{Constructing and Searching for Cover Codes}
To construct an $h$-fold cover code from the base code over a group $G$, we begin by identifying a suitable covering group $\tilde{G}$. Specifically, we systematically search the GAP Small Groups library for a group $\tilde{G}$ of order $|\tilde{G}| = h|G|$ with a normal subgroup $H \trianglelefteq \tilde{G}$ such that $\tilde{G}/H \cong G$.
 To reduce the size of the search space for cover codes, we can first assume that the base group algebra elements contain the identity $e_G \in G$ by Lemma \ref{lemma:equivalence1}. Consequently, by the following proposition, we can assume without loss of generality that the lifted group algebra elements contain the identity $e_{\tilde{G}} \in \tilde{G}$.

\begin{proposition}\label{prop:GeneralCoverEquivalence}
    Let $Q_{\tilde{G}}(\ah_{\text{cover}},\bh_{\text{cover}})$ be any valid cover code for a binary base code $Q_G(\ah_{\mathrm{base}},\bh_{\mathrm{base}})$ over the group-subgroup pair $(\tilde{G},H)$ with the isomorphism $\phi:\tilde{G}/H\to G$. Let $\pi: \mathbb{F}_2[\tilde{G}] \to \mathbb{F}_2[\tilde{G}/H]$ be the canonical projection homomorphism. For any elements $\tilde{g}_a, \tilde{g}_b \in \tilde{G}$, the shifted code $Q_{\tilde{G}}( \ah_{\text{cover}}*\tilde{g}_a, \tilde{g}_b*\bh_{\text{cover}} )$ is an equivalent cover code. Furthermore, this shifted cover code covers the base code $Q_G( \ah_{\mathrm{base}} *g_a, g_b*\bh_{\mathrm{base}} )$, which is equivalent to the original base code, where $g_a = \phi(\pi(\tilde{g}_a))$ and $g_b = \phi(\pi(\tilde{g}_b))$.
\end{proposition}
\begin{proof}
    By Lemma \ref{lemma:equivalence1} applied to the group $\tilde{G}$, multiplying the polynomials by any group elements $\tilde{g}_a$ and $\tilde{g}_b$ yields the equivalent code $Q_{\tilde{G}}(  \ah_{\text{cover}}*\tilde{g}_a, \tilde{g}_b * \bh_{\text{cover}} )$. 
    To determine the base code covered by this new equivalent code, we use the assumptions in Theorem \ref{theorem:coverCodes}:
    \begin{align*}
        \phi(\pi( \ah_{\text{cover}}*\tilde{g}_a )) &= \phi(\pi(\ah_{\text{cover}}))\cdot \phi(\pi(\tilde{g}_a)) =  \ah_{\mathrm{base}}*g_a , \\
        \phi(\pi(\tilde{g}_b*\bh_{\text{cover}})) &= \phi(\pi(\tilde{g}_b))\cdot \phi(\pi(\bh_{\text{cover}}))  = g_b * \bh_{\mathrm{base}} .
    \end{align*}
Now Lemma \ref{lemma:equivalence1} used for the base group $G$ implies that the new base code $Q_G(\ah_{\mathrm{base}}*g_a , g_b*\bh_{\mathrm{base}})$ is equivalent to the original base code $Q_G(\ah_{\mathrm{base}},\bh_{\mathrm{base}})$.
\end{proof}

To minimize the computational search space for cover codes, we can directly apply Proposition \ref{prop:GeneralCoverEquivalence} to assume, without loss of generality, that the identity element belongs to the support of both lifted elements $\ah_{\text{cover}}\in \mathbb{F}_2[\tilde{G}]$ and $\bh_{\text{cover}}\in \mathbb{F}_2[\tilde{G}]$. If a valid cover code does not contain the identity element, we can simply choose any elements $\tilde{x} \in \operatorname{supp}(\ah_{\text{cover}})$ and $\tilde{y} \in \operatorname{supp}(\bh_{\text{cover}})$ and multiply the group algebra elements by the inverses, $\tilde{x}^{-1}$ and $\tilde{y}^{-1}$. This multiplication guarantees that the identity of the cover group is in the support of the new equivalent cover code, which covers an equivalent base code containing the identity of the base group. Thus, fixing the identity element during the search does not omit any unique cover codes up to isomorphism.

Consequently, by fixing one element in both lifted group algebra elements, the search space for an $h$-fold cover of a base code $Q_G(\ah,\bh)$ is reduced to exactly $h^{w_{\ah}+w_{\bh}-2}$ possible codes, once the cover group, the normal subgroup, and the isomorphism $\tilde{G}/H \cong G$ are fixed.  As observed in \cite{symons2025sequencesbivariatebicyclecodes}, this represents a significant reduction from set of all possible normalized codes over $\tilde{G}$ with the same weights, whose size is $\binom{|\tilde{G}|-1}{w_{\ah}-1}\binom{|\tilde{G}|-1}{w_{\bh}-1}$. 

 Using the strategy explained above, we performed a search for the cover codes of certain previously known 2BGA base codes over various groups, including non-abelian ones. We note that our choice of base codes serves to illustrate the construction rather than exhaust most possibilities. Table \ref{tab:base_and_cover_codes} presents these results, highlighting several cover codes with notable parameters. For example, the $[[204,22,\leq 17]]$ code is a double cover of the $[[102,22,9]]$ GB code mentioned in \cite{tripier2026faulttolerantquantumcomputingtrapped}, while the $[[306,22,\leq 24]]$ code is its triple cover. This illustrates the computational advantage of the cover-graph-based method: while an unconstrained search for the $[[204,22,\leq 17]]$ code would require enumerating %1,887,054,437,401 
$1.8\times 10^{12}$ code instances, the search space for the double cover is reduced to only 64 combinations (calculated as $h^{w_{\ah}+w_{\bh}-2} = 2^{4+4-2} = 64$) once the cover group and the lift subgroup are fixed. 

We further note that our group-theoretic approach provides a foundation for analytically studying the lifted logical operators, offering a pathway to explicitly relate the parameters of the base code and cover codes. 

\begin{table*}[htpb]
    \centering
    \caption{Base codes and their cover codes. The sequences ($\ah_{\text{cover}}$, $\bh_{\text{cover}}$, $\ah_{base}$, $\bh_{base}$) are given as integer indices corresponding to the ordered list generated by the GAP 4.14.0 function \texttt{Elements(G)} for their respective groups. The Base GAP Identifier is given as $(l, m)$ via the command \texttt{SmallGroup(l, m)}. The Cover GAP Identifier is given as a tuple $(l, m, s)$, where $l$ and $m$ define the cover group $\tilde{G}$, and $s$ represents the index of the specific normal lifted subgroup $H$ chosen from the ordered list produced by the command \texttt{Filtered(AllSubgroups(G), h -> IsNormal(G,h))[s]}. Distances reported with a $\leq$ symbol represent upper bounds obtained using \texttt{QDistRnd} with $10^6$ iterations. All other distances are exact.}
    \label{tab:base_and_cover_codes}
    \resizebox{\textwidth}{!}{%
    \begin{tabular}{clccc clcccc}
        \toprule
        \multicolumn{5}{c}{Base Code} & \multicolumn{6}{c}{Cover Code} \\
        \cmidrule(lr){1-5} \cmidrule(lr){6-11}
        $[[n, k, d]]$ & Structure & $\ah_{base}$ & $\bh_{base}$ & GAP Id. & $[n, k, d]$ & Structure & Lift $H$ & $\ah_{cover}$ & $\bh_{cover}$ & GAP Id. \\
         & & & & $(l, m)$ & & & & & & $(l, m, s)$ \\
        \midrule
        % First Base Code block spanning 2 rows
        \multirow{4}{*}{$[[40, 8, 5]]$} & \multirow{4}{*}{$C_5 \rtimes C_4$} & \multirow{4}{*}{[1, 2, 4, 10]} & \multirow{4}{*}{[1, 3, 4, 11]} & \multirow{4}{*}{20, 1}
        & $[[80, 8, 10]]$ & $C_5 \rtimes C_8$ & $C_2$ & [1, 2, 5, 23] & [1, 9, 11, 25] & 40, 1, 2 \\
        & & & & & $[[80, 10, 8]]$ & $C_2 \times (C_5 \rtimes C_4)$ & $C_2$ & [1, 2, 10, 22] & [1, 4, 10, 25] & 40, 7, 3 \\
        & & & & & $[[120, 16, 10]]$ & $C_3 \times (C_5 \rtimes C_4)$ & $C_3$ & [1, 6, 20, 28] & [1, 10, 5, 43] & 60, 2, 3 \\
        & & & & & $[[160, 10, \leq16]]$ & $C_2 \times (C_5 \rtimes C_8)$ & $C_2\times C_2$ & [1, 2, 13, 39] & [1, 4, 16, 42] & 80, 9, 6 \\
        \midrule
        % Second Base Code block spanning 3 rows
        \multirow{5}{*}{$[[62, 12, 7]]$} & \multirow{5}{*}{$C_{31}$} & \multirow{5}{*}{[1, 2, 3, 7]} & \multirow{5}{*}{[1, 2, 13, 18]} & \multirow{5}{*}{31, 1}
        & $[[124, 12, 14]]$ & $C_{62}$ & $C_2$ & [1, 3, 5, 14] & [1, 4, 26, 36] & 62, 2, 2 \\
        & & & & & $[[124, 14, 12]]$ & $C_{62}$ & $C_2$ & [1, 4, 6, 13] & [1, 3, 26, 36] & 62, 2, 2\\
        & & & & & $[[186, 12, \leq 19]]$ & $C_{93}$ & $C_3$ & [1, 5, 8, 22] & [1, 7, 40, 55] & 93, 2, 2 \\
        & & & & & $[[248, 18, \leq19]]$ & $C_{124}$ & $C_4$ & [1, 10, 14, 23] & [1, 5, 49, 67] & 124, 2, 3 \\
        & & & & & $[[248, 16, \leq21]]$ & $C_{62}\times C_2$ & $C_2\times C_2$ & [1, 6, 13, 27] & [1, 6, 50, 68] & 124, 4, 5 \\
        \midrule
        % Third Base Code block spanning 3 rows
        \multirow{3}{*}{$[[56, 12, 4]]$} & \multirow{3}{*}{$C_{14}\times C_2$} & \multirow{3}{*}{[1, 4, 12]} & \multirow{3}{*}{[1, 6, 15]} & \multirow{3}{*}{28, 4}
        & $[[112, 12, 8]]$ & $C_{28}\times C_2$ & $C_2$ & [1, 12, 28] & [1, 7, 25] & 56, 8, 2 \\
        & & & & & $[[168, 16, 10]]$ & $C_{14}\times S_3$ & $C_3$ & [1, 12, 48] & [1, 7, 44] & 84, 13, 3\\
        & & & & & $[[280, 12, 16]]$ & $C_{14}\times D_{10}$ & $C_5$ & [1, 4, 38] & [1, 7, 52] & 140, 9, 3\\
        \midrule
        % Fourth Base Code block spanning 3 rows
        \multirow{2}{*}{$[[48, 6, 6]]$} & \multirow{2}{*}{$C_{3}\rtimes C_8$} & \multirow{2}{*}{[1, 2, 3, 14]} & \multirow{2}{*}{[1, 2, 14, 17]} & \multirow{2}{*}{24, 1}
        & $[[96, 12, 10]]$ & $C_{2}\times(C_3 \rtimes C_8)$ & $C_2$ & [1, 2, 11, 22] & [1, 2, 33, 29] & 48, 9, 3 \\
        & & & & & $[[192, 14, \leq 16]]$ & $C_{3}\rtimes ((C_8\times C_2)\rtimes C_2)$ & $C_2\times C_2$ & [1, 2, 36, 30] & [1, 2, 57, 64] & 96, 37, 9\\
        \midrule
        % Fifth Base Code block spanning 3 rows
        \multirow{3}{*}{$[[36, 4, 6]]$} & \multirow{3}{*}{$(C_{3}\times C_3)\rtimes C_2$} & \multirow{3}{*}{[1, 2, 3, 6]} & \multirow{3}{*}{[1, 2, 8, 15]} & \multirow{3}{*}{18, 4}
        & $[[72, 8, 9]]$ & $(C_{3}\times C_3)\rtimes C_4$ & $C_2$ & [1, 2, 4, 8] & [1, 2, 12, 27] & 36, 7, 2 \\
        & & & & & $[[108, 12, 9]]$ & $C_{3}\times ((C_3\times C_3)\rtimes C_2)$ & $C_3$ & [1, 6, 4, 17] & [1, 15, 13, 50] & 54, 13, 6\\
        & & & & & $[[144, 10, 14]]$ & $(C_3\times C_3)\rtimes C_8$ & $C_4$ & [1, 2, 14, 23] & [1, 8, 33, 45] & 72, 13, 7\\
        \midrule
        % Sixth Base Code block spanning 3 rows
        \multirow{3}{*}{$[[70, 16, 6]]$} & \multirow{3}{*}{$C_{35}$} & \multirow{3}{*}{[1, 3, 9, 19]} & \multirow{3}{*}{[1, 3, 13, 23]} & \multirow{3}{*}{35, 1}
        & $[[140, 16, 12]]$ & $C_{70}$ & $C_2$ & [1, 4, 15, 39] & [1, 6, 23, 43] & 70, 4, 2 \\
        & & & & & $[[210, 20, \leq16]]$ & $C_{105}$ & $C_3$ & [1, 7, 18, 54] & [1, 12, 29, 77] & 105, 2, 2\\
        & & & & & $[[280, 22, \leq18]]$ & $C_{70}\times C_2$ & $C_2\times C_2$ & [1, 8, 31, 74] & [1, 10, 63, 93] & 140, 11, 5\\
        \midrule
        % Seventh Base Code block spanning 3 rows
        \multirow{4}{*}{$[[84, 16, 8]]$} & \multirow{4}{*}{$C_{7}\times S_3$} & \multirow{4}{*}{[1, 2, 3, 29]} & \multirow{4}{*}{[1, 4, 15, 27]} & \multirow{4}{*}{42, 3}
        & $[[168, 20, 14]]$ & $C_{7}\times (C_3\rtimes C_4)$ & $C_2$ & [1, 2, 3, 51] & [1, 5, 36, 60] & 84, 3, 2 \\
        & & & & & $[[252, 16, \leq 21]]$ & $C_{21}\times S_3$ & $C_3$ & [1, 6, 21, 71] & [1, 11, 57, 78] & 126, 12, 3\\
        & & & & & $[[252, 20, \leq18]]$ & $C_{21}\times S_3$ & $C_3$ & [1, 15, 21, 71] & [1, 22, 57, 63] & 126, 12, 3\\
        & & & & & $[[252, 24, \leq17]]$ & $C_7\times((C_3\times C_3)\rtimes C_2)$ & $C_3$ & [1, 2, 11, 65] & [1, 4, 43, 98] & 126, 14, 2\\
        \midrule
        % Eight Base Code block spanning 3 rows
        \multirow{5}{*}{$[[42, 8, 6]]$} & \multirow{5}{*}{$C_{21}$} & \multirow{5}{*}{[1, 3, 6, 14]} & \multirow{5}{*}{[1, 3, 8, 12]} & \multirow{5}{*}{21, 2}
        & $[[84, 8, 12]]$ & $C_{42}$ & $C_2$ & [1, 6, 12, 29] & [1, 6, 14, 21] & 42, 6, 2 \\
        & & & & & $[[126, 8, 15]]$ & $C_{21}\times C_3$ & $C_3$ & [1, 4, 15, 35] & [1, 12, 23, 33] & 63, 4, 2\\
        & & & & & $[[126, 20, 9]]$ & $C_{63}$ & $C_3$ & [1, 18, 8, 31] & [1, 3, 23, 45] & 63, 2, 2\\
        & & & & & $[[168, 10, \leq18]]$ & $C_{84}$ & $C_4$ & [1, 4, 24, 45] & [1, 7, 27, 48] & 84, 6, 4\\
        & & & & & $[[168, 12, \leq17]]$ & $C_{84}$ & $C_4$ & [1, 7, 24, 64] & [1, 13, 21, 48] & 84, 6, 4\\
        & & & & & $[[168, 14, \leq16]]$ & $C_{42}\times C_2$ & $C_2\times C_2$ & [1, 10, 27, 56] & [1, 15, 23, 42] & 84, 15, 6\\
        \midrule
        % Ninth Base Code block spanning 3 rows
        \multirow{3}{*}{$[[62, 10, 7]]$} & \multirow{3}{*}{$C_{31}$} & \multirow{3}{*}{[1, 2, 13]} & \multirow{3}{*}{[1, 2, 3, 28]} & \multirow{3}{*}{31, 1}
        & $[[124, 10, 12]]$ & $C_{62}$ & $C_2$ & [1, 3, 26] & [1, 4, 5, 56] & 62, 2, 2 \\
        & & & & & $[[186, 10, \leq17]]$ & $C_{93}$ & $C_3$ & [1, 5, 38] & [1, 3, 10, 83] & 93, 2, 2\\
        & & & & & $[[248, 10, \leq 21]]$ & $C_{62}\times C_2$ & $C_2\times C_2$ & [1, 6, 51] & [1, 9, 10, 113] & 124, 4, 5\\
        \midrule
        % Tenth Base Code block spanning 3 rows
        \multirow{3}{*}{$[[42, 10, 5]]$} & \multirow{3}{*}{$C_{21}$} & \multirow{3}{*}{[1, 5, 13]} & \multirow{3}{*}{[1, 3, 10, 19]} & \multirow{3}{*}{21, 2}
        & $[[84, 10, 9]]$ & $C_{42}$ & $C_2$ & [1, 8, 28] & [1, 6, 22, 37] & 42, 6, 2 \\
        & & & & & $[[126, 10, 13]]$ & $C_{63}$ & $C_3$ & [1, 6, 38] & [1, 3, 29, 46] & 63, 2, 2\\
        & & & & & $[[210, 10, \leq 19]]$ & $C_{105}$ & $C_5$ & [1, 14, 76] & [1, 41, 34, 102] & 105, 2, 3\\
        \midrule
        % Eleventh Base Code block spanning 3 rows
        \multirow{5}{*}{$[[102, 22, 9]]$} & \multirow{5}{*}{$C_{51}$} & \multirow{5}{*}{[1, 14, 45, 35]} & \multirow{5}{*}{[1, 27, 32, 50]} & \multirow{5}{*}{51, 1}
        & $[[204, 22, \leq 17]]$ & $C_{102}$ & $C_2$ & [1, 26, 87, 71] & [1, 51, 65, 100] & 102, 4, 2 \\
        & & & & & $[[204, 24, 15]]$ & $C_{102}$ & $C_2$ & [1, 26, 87, 68] & [1, 54, 62, 100] & 102, 4, 2\\
        & & & & & $[[204, 28, 12]]$ & $C_{102}$ & $C_2$ & [1, 26, 87, 68] & [1, 51, 65, 100] & 102, 4, 2\\
        & & & & & $[[306, 22, \leq 24]]$ & $C_{153}$ & $C_3$ & [1, 31, 134, 104] & [1, 80, 105, 139] & 153, 1, 2\\
        & & & & & $[[306, 26, \leq18]]$ & $C_{51}\times C_{3}$ & $C_3$ & [1, 35, 138, 110] & [1, 84, 89, 151] & 153, 2, 2\\
        \midrule
        % Twelve Base Code block spanning 3 rows
        \multirow{2}{*}{$[[126,20,11]]$} & \multirow{2}{*}{$C_{63}$} & \multirow{2}{*}{[1, 6, 26, 60]} & \multirow{2}{*}{[1, 22, 53, 35]} & \multirow{2}{*}{63, 2}
        & $[[252, 20, \leq 21]]$ & $C_{126}$ & $C_2$ & [1, 10, 53, 122] & [1, 40, 98, 71] & 126, 6, 2 \\
        & & & & & $[[252, 22, \leq18]]$ & $C_{126}$ & $C_2$ & [1, 16, 44, 122] & [1, 40, 106, 71] & 126, 6, 2\\
        \bottomrule
    \end{tabular}%
    }
\end{table*}

\section{Conclusion and Future Directions}

In this paper, we introduced a family of quantum LDPC codes constructed using group actions on the cosets of a finite group. Our construction generalizes the previously known 2BGA family of quantum LDPC codes, which includes BB codes and GB codes. Through targeted computer searches, we identified new weight-6 and weight-8 quantum LDPC codes whose parameters are on par with the state of the art. We developed a maximally packed syndrome extraction schedule that utilizes a circuit depth of $w+2$ for a code with stabilizer weight at most $w$. We evaluated the circuit-level performance of our codes using various decoders and demonstrated that they achieve logical error rates comparable to those of previously known, high-performing quantum LDPC codes under the standard noise models. Furthermore, we presented a framework to generate sequences of 2BGA codes from covering graphs.

Our work opens several new directions for future research:
\begin{itemize}
    \item \textbf{Hardware-specific implementation:} While our syndrome extraction schedule is maximally packed and optimizes time-domain depth, further investigation into the physical layout and qubit connectivity requirements for specific quantum hardware architectures is a natural next step. For instance, exploring coset spaces that admit a bi-planar toric layout (similar to BB codes) but whose long-range connections extend beyond a single cyclic shift is an interesting new direction for superconducting platforms. Similarly, identifying coset spaces that admit efficient scheduling tailored for neutral-atom or trapped-ion architectures represents another promising direction.
    \item \textbf{Exploring logical operators:} While our work focuses on the construction and syndrome extraction of these codes, a detailed investigation of their logical operators remains an open problem. Explicitly constructing logical bases and finding minimum-weight logical operators is a natural next step toward practical implementation.
    \item \textbf{Distance bounds on graph cover codes:} Another interesting direction is to study bounds on the distances of cover codes, given the base code parameters, by restricting the base and cover groups from Section~\ref{sec:CoverGraphs} to specific, highly structured families like cyclic groups. While this problem was partially addressed in \cite{symons2025sequencesbivariatebicyclecodes} for BB codes, we believe that our generalized group-theoretic framework could be leveraged to extend these results.
\end{itemize}

\section{Data Availability}
The scripts and source code to reproduce the data generated in this paper are publicly available at \cite{aydin_coset2bga_2026}.

\begin{acknowledgements}

We are grateful to Nicolas Delfosse, Leonid Pryadko, and  Min Ye for stimulating discussions. The research of A.A. and A.B. was partially supported by NSF grants CCF2330909 and CCF2526035.  We also acknowledge the University of Maryland supercomputing resources (\href{http://hpcc.umd.edu}{http://hpcc.umd.edu}) made available for conducting the research reported in this paper.
\end{acknowledgements}

\FloatBarrier
\bibliography{TwoBGA}

\appendix
\section{Additional Proofs \label{appendix:additionalproofs}}

\subsection{Proof of Proposition \ref{prop:LRhomomorpshim}}
 First, let us prove that $L(g)$ permutes the cosets in $G/H$ for all $g\in G$. Let $g\in G$. The following chain of relationships
    \begin{align*}
            xH=yH&\iff x^{-1}y\in H \\ &\iff x^{-1}y=x^{-1}g^{-1}gy =(gx)^{-1}(gy)\in H \\ &\iff (gx)H=(gy)H \\&\iff L(g)(xH)=L(g)(yH),
    \end{align*}
shows that $L(g)$ is well defined by reading it from left to right, and it is an injection by reading it from right to left. Hence, $L(g)$ is a well-defined left action of $g\in G$ on $G/H$.  It remains to show that $L$ commutes with the group product. For $g_1,g_2\in G$, we have
\begin{align*}
    L(g_1)\circ L(g_2)(xH) &= L(g_1)((g_2x)H)\\
    &=(g_1(g_2x))H\\
    &=((g_1g_2)x)H\\
    &=L(g_1g_2)(xH),
\end{align*}
as required.

We repeat a similar sequence of steps for $R$. Let $g\in N_G(H)$, then
\begin{align*}
    xH=yH &\iff x^{-1}y\in H\\
    &\iff g^{-1}x^{-1}yg \in g^{-1}Hg=H\\
    &\iff (xg)^{-1}(yg)\in H\\
    &\iff (xg)H = (yg)H\\
    &\iff R(g)(xH)=R(g)(yH),
\end{align*}
where we used $g^{-1}Hg=H$ since $g\in N_G(H)$. This proves that $R(g)$ is a well-defined right action on left cosets of $H$ in $G$. Furthermore, taking $g_1,g_2\in N_G(H)$,
\begin{align*}
    R(g_1)\circ R(g_2) (xH) &= R(g_1)((xg_2)H)\\
    &= ((xg_2)g_1)H\\
    &=(x(g_2g_1))H\\
    &=R(g_2g_1)(xH),
\end{align*}
proving that $R(g_1)\circ R(g_2) = R(g_2g_1)$.

\subsection{Proof of Proposition \ref{prop:commutativity}}
Let $xH$ be any left coset of $H$ in $G$. Then
\begin{align*}
    L(g_1)\circ R(g_2)(xH) &=L(g_1)((xg_2)H)\\
    &= (g_1(xg_2))H\\
    &= ((g_1x)g_2)H\\
    &=R(g_2)((g_1x)H)\\
    &= R(g_2)\circ L(g_1)(xH)
\end{align*}
\subsection{Proof of Proposition \ref{prop:ImageLandR}}
\cite[Thm.4.2.3]{Dummit2004}

Once we show that $\ker L=\operatorname{Core}_G(H)$, the claim
$\operatorname{Im} L\cong{G}/{\ker L}$ will follow by the first isomorphism theorem. We have, for all $x\in G$, 
\begin{align*}
   g\in \ker L &\iff L(g)(xH) = xH\\
   &\iff (gx)H=xH\\
   & \iff x^{-1}(gx) \in H\\
   & \iff g \in xHx^{-1}\\
   & \iff g \in \bigcap_{x\in G} xHx^{-1}\\
   &\iff g \in \operatorname{Core}_G(H).
\end{align*}
Similarly, $\ker R=H$ since
\begin{align*}
    g\in \ker R  &\iff R(g)(xH)=xH\\
    &\iff (xg)H = xH\\
    & \iff x^{-1}(xg) \in H\\
    &\iff g \in H.
\end{align*}

\section{Description of the Wheel Graph in Lemma \ref{lemma:thickness}}\label{appendix:wheel graph}

Let $\mathcal{T}_i$ be the bipartite subgraph defined in Lemma \ref{lemma:thickness} by the parity-check matrices
\begin{align*}
        &H_X^{(i)}=[\bfL(g_r)+\bfL(g_s)\mid \bfR(g_t)],\\
        &H_Z^{(i)}=[\bfR(g_t^{-1})\mid \bfL(g_r^{-1})+\bfL(g_s^{-1})]
\end{align*} 
For a code of even length $n$, let $l=n/2$. We assume that the assignments of the data qubits $D_L[i], D_R[i]$ and the ancilla qubits $X[i], Z[i]$ for $i=1,2,\ldots,l$ are those described in Section \ref{sec:SC_Circuit}. Let $M$ be a permutation matrix, and define the mapping $M(i) = j$ if $M_{i,j} = 1$. Note that for any two permutation matrices $M$ and $P$ of the same size, if $M$ and $P$ commute, then $M(P(i)) = P(M(i))$. For better readability, we write $MP(i)$ for $M(P(i))$. Using this notation and the fact that all $\bfL$ matrices commute with all $\bfR$ matrices, it follows that every check qubit $X[i]$ belongs to the connected planar component shown in Figure~\ref{fig:wheelgraph}. Following \cite{Bravyi2024}, we call such graphs 'wheel graphs.'

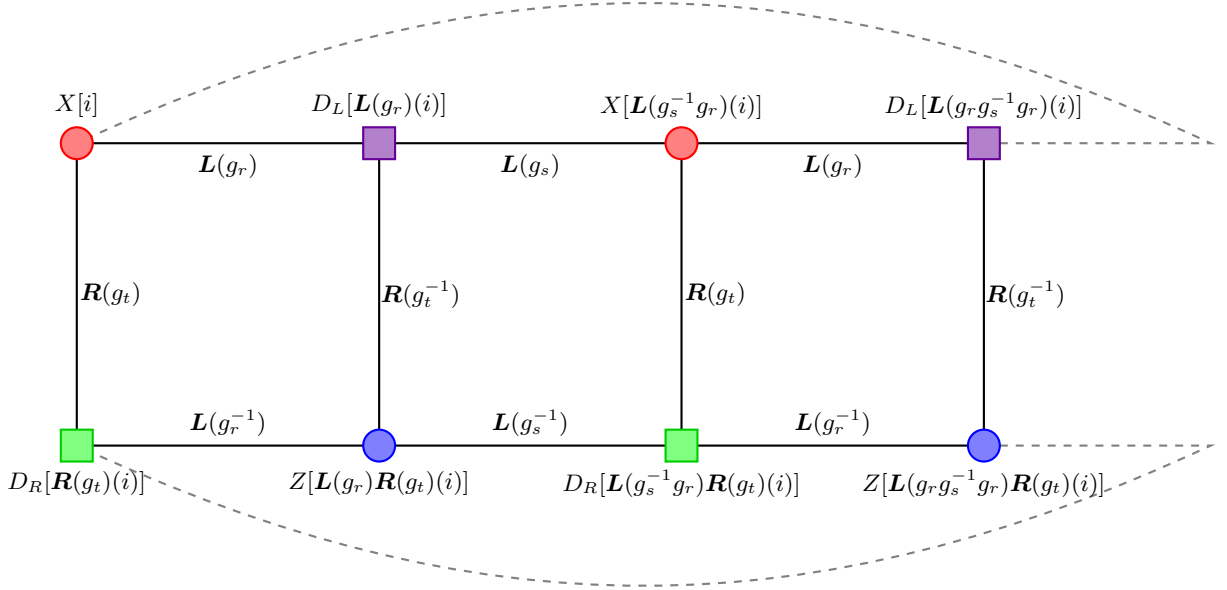
\begin{figure*}
    \centering
    \caption{The description of the wheel graph in Lemma \ref{lemma:thickness}.}
    \label{fig:wheelgraph}
\begin{tikzpicture}[
    redcircle/.style={draw=red, thick, fill=red!50, circle, minimum size=12pt, inner sep=0pt},
    bluecircle/.style={draw=blue, thick, fill=blue!50, circle, minimum size=12pt, inner sep=0pt},
    greensquare/.style={draw=green!80!black, thick, fill=green!50, rectangle, minimum size=12pt, inner sep=0pt},
    purplesquare/.style={draw=blue!60!red, thick, fill=blue!60!red!50, rectangle, minimum size=12pt, inner sep=0pt}
    ]

    % 1. Draw the two horizontal side rails (Width from x=0 to x=12)
    % The top rail is now at y=4 to make the height 4 units
    \draw[thick] (0,4) -- (12,4); 
    \draw[thick] (0,0) -- (12,0); 

% 2. Draw the wrapping dashed lines (Periodic Boundaries)
    % Top line: Extends to x=13.5, arcs UP over the text, lands at x=-1.5, enters x=0
    \draw[thick, dashed, gray] (12,4) -- (15,4) to[bend right=25] (0,4) -- (0,4);
    
    % Bottom line: Extends to x=13.5, arcs DOWN under the text, lands at x=-1.5, enters x=0
    \draw[thick, dashed, gray] (12,0) -- (15,0) to[bend left=25] (0,0) -- (0,0);
    % ---------------------------------------------------------
    % --- Column 1 (Leftmost rung at x = 0) ---
    \draw[thick] (0,4) -- (0,0);
    \node[redcircle] at (0,4) {};
    \node at (0, 4.5) {$X[i]$}; 
    \node at (0.45, 2.0) {$\bfR(g_t)$};
    \node[greensquare] at (0,0) {};
    \node at (0, -0.5) {$D_R[\bfR(g_t)(i)]$}; 
    \node at (2.0, 3.70) {$\bfL(g_r)$};
    \node at (2.0, 0.30) {$\bfL(g_r^{-1})$};
    % ---------------------------------------------------------
    % --- Column 2 (Second rung at x = 4) ---
    \draw[thick] (4,4) -- (4,0);
    \node[purplesquare] at (4,4) {};
    \node at (4, 4.5) {$D_L[\bfL(g_r)(i)]$};  
    \node at (4.55, 2.0) {$\bfR(g_t^{-1})$};
    \node[bluecircle] at (4,0) {};
    \node at (4, -0.5) {$Z[\bfL(g_r)\bfR(g_t)(i)]$};
    \node at (6.0, 3.70) {$\bfL(g_s)$};
    \node at (6.0, 0.30) {$\bfL(g_s^{-1})$};
    % ---------------------------------------------------------
    % --- Column 3 (Third rung at x = 8) ---
    \draw[thick] (8,4) -- (8,0);
    \node[redcircle] at (8,4) {};
    \node at (8, 4.5) {$X[\bfL(g_s^{-1}g_r)(i)]$};
    \node at (8.45, 2.0) {$\bfR(g_t)$};
    \node[greensquare] at (8,0) {};
    \node at (8, -0.5) {$D_R[\bfL(g_s^{-1}g_r)\bfR(g_t)(i)]$};
    \node at (10.0, 3.70) {$\bfL(g_r)$};
    \node at (10.0, 0.30) {$\bfL(g_r^{-1})$};
    % ---------------------------------------------------------
    % --- Column 4 (Fourth rung at x = 12) ---
    \draw[thick] (12,4) -- (12,0);
    \node[purplesquare] at (12,4) {};
    \node at (12, 4.5) {$D_L[\bfL(g_rg_s^{-1}g_r)(i)]$};  
    \node at (12.55, 2.0) {$\bfR(g_t^{-1})$};
    \node[bluecircle] at (12,0) {};
    \node at (12, -0.5) {$Z[\bfL(g_rg_s^{-1}g_r)\bfR(g_t)(i)]$};

\end{tikzpicture}
\end{figure*}

\section{Proof of Proposition \ref{prop:SyndromeCircuit}}\label{appendix:SC_Circuit_Proof}
We first establish that each syndrome extraction cycle in 
Algorithm~\ref{alg:syndrome_extraction} requires exactly $w_{\ah} + w_{\bh} + 2$ 
time steps for any code $Q_G^H(\ah, \bh)$. To this end, observe that 
\texttt{CnotSchedule} produces a circuit consisting of 
$|\vec{g}_X| + |\vec{h}_X| + |\vec{f}_X|$ CNOT rounds, corresponding to 
its three sequential loops. 
The proof depends on the parity of $w_{\ph}$. We will start with the even case.  Then the state initializations at Line~\ref{line:InitWpEven} occupy a single time step. 
The call to \texttt{CnotSchedule} then requires $w_{\ph} + w_{\uh }$ CNOT rounds, 
since $|\vec{g}_X| + |\vec{h}_X| = w_{\ph}$ and $|\vec{f}| = w_{\uh }$ by 
construction. Finally, the ancilla measurements contribute one additional 
time step, yielding a total cycle depth of 
$w_{\ph} + w_{\uh } + 2 = w_{\ah} + w_{\bh} + 2 = w + 2$. Now suppose that $w_{\ph}$ is odd, then the state initializations on Line~\ref{line:InitWpOdd} again occupy a single time step. 
The subsequent call to \texttt{CnotSchedule} requires $w_{\ph} + w_{\uh } - 1$ 
CNOT rounds, since $|\vec{g}_X| + |\vec{h}_X| = w_{\ph} - 1$ and 
$|\vec{f}| = w_{\uh }$ by construction. As shown in 
Lines~\ref{line:MZWpOdd}--\ref{line:MXWpOdd}, the CNOT associated with the isolated element 
$e_x$ (resp.\ $e_z$) is executed concurrently with the initialization 
(resp.\ measurement) step, contributing two time step, yielding 
a total cycle depth of $(w_{\ph} +w_{\uh } - 1) + 3 = w_{\ah} + w_{\bh} + 2 = w + 2$. 

To prove correctness of the syndrome extraction cycle, we employ the 
stabilizer tableau formalism~\cite{PhysRevA.70.052328}, following the 
approach of Ref.~\cite{Bravyi2024}. Recall that the action of the CNOT 
gate on Pauli operators is given by
\begin{align*}
    X \otimes I &\rightarrow X \otimes X, & I \otimes X &\rightarrow I \otimes X, \\
    Z \otimes I &\rightarrow Z \otimes I, & I \otimes Z &\rightarrow Z \otimes Z.
\end{align*}
Since this action does not mix $X$-type and $Z$-type Pauli operators, 
the evolution of $X$-type and $Z$-type stabilizers can be tracked 
independently. Let us first examine the evolution of $X$-type stabilizers. Suppose that $Q_G^H(\ah, \bh)$ is the CSS code given by Construction~\ref{cons:Generalized2BGA} and let $l = [G:H] = n/2$, where $n$ is the code length. To introduce the stabilizer tableau, we define a $2l \times 4l$ binary matrix whose $j$-th column represents the qubit
\begin{align*}
    Q[j] = \begin{cases}
        X[j]        & \text{if } 1 \leq j \leq l, \\
        D_L[j-l]    & \text{if } l < j \leq 2l, \\
        D_R[j-2l]   & \text{if } 2l < j \leq 3l, \\
        Z[j-3l]     & \text{if } 3l < j \leq 4l.
    \end{cases}
\end{align*}
The first $l$ rows track the single-weight check operators acting on the 
ancilla qubits $X[i]$ for $i = 1, 2, \ldots, l$. Since these qubits are 
initialized to the $\ket{+}$ state, the first $l$ rows of the stabilizer 
tableau take the block form $(I,\, 0,\, 0,\, 0)$, where $I := I_l$ denotes 
the $l \times l$ identity matrix. The remaining $l$ rows track the 
$X$-type stabilizer generators of the CSS code $Q_G^H(\ah, \bh)$. Therefore, at the beginning of each syndrome extraction cycle, the 
stabilizer tableau takes the following block form:
\begin{align*}
    \begin{bmatrix}
        I & 0 & \bf{0} & 0\\
        0 & \bfL(\ah) & \bfR(\bh) & 0
    \end{bmatrix}.
\end{align*}
A valid syndrome extraction cycle must leave the $X$-type stabilizer 
generators of the code invariant, while updating the ancilla stabilizers 
to reflect the $X$-type check operators, so that measuring the ancilla 
qubits in the $X$-basis yields the syndrome. In terms of the stabilizer 
tableau, this requires the bottom $l$ rows to remain unchanged, while the 
top $l$ rows evolve from $(I,\, 0,\, 0,\, 0)$ to $(I,\, \bfL(\ah),\, \bfR(\bh),\, 0)$.

To systematically track the algebraic propagation of operators throughout the 
circuit, it is helpful to formalize how simultaneously applied CNOT gates transform these binary blocks. Let $\bfC$ and $\bfT$ be $l \times l$ binary matrices over $\mathbb{F}_2$, where each row denotes an $X$-type Pauli operator on $l$ qubits, mapping $0$ to the identity operator $I$ and $1$ to the Pauli $X$ operator. Let 
$\bfM$ be an $l \times l$ permutation matrix that defines CNOT gates applied within a single time step. Specifically,  if $M_{i,j} = 1$, a CNOT gate is executed with the $i$-th qubit of the block ($\bfC$) as the control and the $j$-th qubit of the block ($\bfT$) as the target. Recalling that the CNOT gate acts on $X$-type operators as $X \otimes I \rightarrow X \otimes X$ and $I \otimes X \rightarrow I \otimes X$, the collective action of these parallel gates updates the stabilizer blocks according to the matrix transformation:
\begin{align*}
    (\bfC, \bfT) \xrightarrow{\bfM} (\bfC, \bfT + \bfC\bfM).
\end{align*}

Recall that Algorithm~\ref{alg:syndrome_extraction} admits two possible 
dynamic assignments. If $\ph = \ah$, then 
$(D_p, D_u, \bfM_{\ph}, \bfM_{\uh }) \leftarrow (D_L, D_R, \bfL, \bfR)$, and if $\ph = \bh$, then $(D_p, D_u, \bfM_{\ph}, \bfM_{\uh }) \leftarrow (D_R, D_L, \bfR, \bfL)$. In either case, the $X$-type check operators act on the data qubits in block $D_p$ via the permutation matrices $\bfM_{\ph}(\ph)$, and act on the data qubits in block $D_u$ via the permutation matrices $\bfM_{\uh }(\uh )$. Similarly, the $Z$-type check operators act on block $D_p$ via $\bfM_{\uh }(\uh )^T$ and on block $D_u$ via $\bfM_{\ph}(\ph)^T$. Therefore, it suffices to show that Algorithm~\ref{alg:syndrome_extraction} 
yields the tableau transformation
\begin{align*}
    \begin{bmatrix}
        I & 0 & 0 & 0\\
        0 & \bfM_{\ph}(\ph) & \bfM_{\uh }(\uh ) & 0
    \end{bmatrix} 
    \longrightarrow  
    \begin{bmatrix}
        I & \bfM_{\ph}(\ph) & \bfM_{\uh }(\uh ) & 0\\
        0 & \bfM_{\ph}(\ph) & \bfM_{\uh }(\uh ) & 0
    \end{bmatrix}.
\end{align*}

In the calculations that follow, for any sequence $\vec{v}$ of elements 
from a group $G_v$, we define the corresponding group algebra element 
over $\mathbb{F}_2(G_v)$ as ${\mathfrak {v}} = \sum_{g \in \vec{v}} g$. 
Furthermore, we omit the index variable $i$ when an operation is 
performed simultaneously for all $i = 1, \ldots, N$. We first establish how the sequence of CNOT operations described in 
Algorithm~\ref{alg:schedule_cnot} transforms the stabilizer tableau. 
To determine the transformation matrix $T$ enacted by the 
\textsc{CnotSchedule} subroutine on an arbitrary input basis 
$[X,\, D_p,\, D_u,\, Z]$, we track the evolution of each operator 
starting from the identity:
\begin{align*}
    \begin{bmatrix}
        I & 0 & 0 & 0\\
        0 & I & 0 & 0\\
        0 & 0 & I & 0\\
        0 & 0 & 0 & I
    \end{bmatrix}.
\end{align*}

\begin{widetext}
\noindent After the loop at Lines~\ref{line:1_to_gX_start}--\ref{line:1_to_gX_end} (for $k = 1$ to $|\vec{g}_X|$):
\begin{align*}
    \xrightarrow{\mathrm{CNOT}(X,\, D_p[\bfM_{\ph}(g_{X,k})])}
    \begin{bmatrix}
        I & \bfM_{\ph}({\gh}_X) & 0 & 0\\
        0 & I & 0 & 0\\
        0 & 0 & I & 0\\
        0 & 0 & 0 & I
    \end{bmatrix}
    \xrightarrow{\mathrm{CNOT}(D_u,\, Z[\bfM_{\ph}(g_{Z,k})])}
    \begin{bmatrix}
        I & \bfM_{\ph}({\gh}_X) & 0 & 0\\
        0 & I & 0 & 0\\
        0 & 0 & I & \bfM_{\ph}({\gh}_Z)\\
        0 & 0 & 0 & I
    \end{bmatrix}.
\end{align*}

\noindent After the updated loop at Lines~\ref{line:1_to_fX_start}--\ref{line:1_to_fX_end} (for $k = 1$ to $|\vec{f}_X|$):
\begin{align*}
    &\xrightarrow{\mathrm{CNOT}(X,\, D_u[\bfM_{\uh }(f_{X,k})])}
    \begin{bmatrix}
        I & \bfM_{\ph}({\gh}_X) & \bfM_{\uh }({\fh}_X) & 0\\
        0 & I & 0 & 0\\
        0 & 0 & I & \bfM_{\ph}({\gh}_Z)\\
        0 & 0 & 0 & I
    \end{bmatrix} \notag \\ 
    &\xrightarrow{\mathrm{CNOT}(D_p,\, Z[\bfM_{\uh }(f_{Z,k})])}
    \begin{bmatrix}
        I & \bfM_{\ph}({\gh}_X) & \bfM_{\uh }({\fh}_X) & \bfM_{\ph}({\gh}_X)\bfM_{\uh }({\fh}_Z)\\
        0 & I & 0 & \bfM_{\uh }({\fh}_Z)\\
        0 & 0 & I & \bfM_{\ph}({\gh}_Z)\\
        0 & 0 & 0 & I
    \end{bmatrix}.
\end{align*}

\noindent After the loop at Lines~\ref{line:1_to_hX_start}--\ref{line:1_to_hX_end} (for $k = 1$ to $|\vec{h}_X|$):
\begin{align*}
    &\xrightarrow{\mathrm{CNOT}(X,\, D_p[\bfM_{\ph}(h_{X,k})])}
    \begin{bmatrix}
        I & \bfM_{\ph}({\gh}_X+{\hh}_X) & \bfM_{\uh }({\fh}_X) & \bfM_{\ph}({\gh}_X)\bfM_{\uh }({\fh}_Z)\\
        0 & I & 0 & \bfM_{\uh }({\fh}_Z)\\
        0 & 0 & I & \bfM_{\ph}({\gh}_Z)\\
        0 & 0 & 0 & I
    \end{bmatrix} \notag \\
    &\xrightarrow{\mathrm{CNOT}(D_u,\, Z[\bfM_{\ph}(h_{Z,k})])}
    \begin{bmatrix}
        I & \bfM_{\ph}({\gh}_X+{\hh}_X) & \bfM_{\uh }({\fh}_X) & \bfM_{\ph}({\gh}_X)\bfM_{\uh }({\fh}_Z) + \bfM_{\uh }({\fh}_X)\bfM_{\ph}({\hh}_Z)\\
        0 & I & 0 & \bfM_{\uh }({\fh}_Z)\\
        0 & 0 & I & \bfM_{\ph}({\gh}_Z+{\hh}_Z)\\
        0 & 0 & 0 & I
    \end{bmatrix}.
\end{align*}

Since $\vec{f}_X$ and $\vec{f}_Z$ are both sequence orderings of the identical support set (${\fh}_X={\fh}_Z$), we can simplify the top-right matrix entry as $\bfM_{\uh }({\fh}_X)\bfM_{\ph}({\gh}_X + {\hh}_Z)$. Thus, the transformation matrix $T$ corresponding to \texttt{CnotSchedule}$(D_p,\, \bfM_{\ph},\, \vec{g}_X,\, \vec{h}_X,\, \vec{g}_Z,\, \vec{h}_Z,\, \vec{f}_X,\, \vec{f}_Z)$ reduces to the block form:
\begin{align*}
    T = \begin{bmatrix}
        I & \bfM_{\ph}({\gh}_X+{\hh}_X) & \bfM_{\uh }({\fh}_X) & 
        \bfM_{\uh }({\fh}_X)\bfM_{\ph}({\gh}_X+{\hh}_Z)\\
        0 & I & 0 & \bfM_{\uh }({\fh})\\
        0 & 0 & I & \bfM_{\ph}({\gh}_Z + {\hh}_Z)\\
        0 & 0 & 0 & I
    \end{bmatrix}.
\end{align*}
\end{widetext}
We now analyze Algorithm~\ref{alg:syndrome_extraction} in each of the 
two cases.

\medskip
\noindent\textbf{Case 1: $w_{\ph}$ odd.}
The CNOT operation at Line~\ref{line:InitWpOdd} acts on the initial tableau as follows: If $(S,U)=(Z,X)$, then:
\begin{align*}
    \xrightarrow{\mathrm{CNOT}(D_u,\, Z[\bfM_{\ph}(g_{z_0})])}
    \begin{bmatrix}
        I & 0 & 0 & 0\\
        0 & \bfM_{\ph}(\ph) & \bfM_{\uh }(\uh ) & 
        \bfM_{\uh }(\uh )\bfM_{\ph}(g_{z_0})
    \end{bmatrix}.
\end{align*}
Applying the transformation $T$ then yields:
\begin{align*}
    \xrightarrow{T}
    \begin{bmatrix}
        I & \bfM_{\ph}({\gh}_X+{\hh}_X) & \bfM_{\uh }(\uh ) & 
        \bfM_{\uh }(\uh )\bfM_{\ph}(2{\gh}_X)\\
        0 & \bfM_{\ph}(\ph) & \bfM_{\uh }(\uh ) & 
        2\bfM_{\ph}(\ph)\bfM_{\uh }(\uh )
    \end{bmatrix},
\end{align*}
where we have used $\ph = {\gh}_Z + {\hh}_Z + g_{z_0}$, 
$\uh  = {\fh}$, and ${\gh}_X = {\hh}_Z$. 
Finally, the CNOT operation at Line~\ref{line:MZWpOdd} gives:
\begin{align*}
    \xrightarrow{\mathrm{CNOT}(X,\, D_p[\bfM_{\ph}(h_{x_0})])}
    \begin{bmatrix}
        I & \bfM_{\ph}(\ph) & \bfM_{\uh }(\uh ) & 0\\
        0 & \bfM_{\ph}(\ph) & \bfM_{\uh }(\uh ) & 0
    \end{bmatrix},
\end{align*}
where we have used $\ph = {\gh}_X + {\hh}_X + h_{x_0}$, together 
with $2\bfM_{\ph}(\ph)\bfM_{\uh }(\uh ) = 0$ and 
$\bfM_{\uh }(\uh )\bfM_{\ph}(2{\gh}_X) = 0$, since all 
computations are performed over the binary field $\mathbb{F}_2$.

\noindent If $(S,U)=(X,Z)$, then:
\begin{align*}
    \xrightarrow{ \text{CNOT}(X, D_p[\bfM_{\ph}(g_{x_0})])}
    \begin{bmatrix}
        I & \bfM_{\ph}(g_{x_0}) & 0 & 0\\
        0 & \bfM_{\ph}(\ph) & \bfM_{\uh }(\uh ) & 
        0
    \end{bmatrix}.
\end{align*}
Applying the transformation $T$ then yields:
\begin{align*}
    \xrightarrow{T}
    \begin{bmatrix}
        I & \bfM_{\ph}(\ph) & \bfM_{\uh }(\uh ) & 
       \bfM_{\uh }(\uh )\bfM_{\ph}({\gh}_X+{\hh}_Z+g_{x_0}) \\\
        0 & \bfM_{\ph}(\ph) & \bfM_{\uh }(\uh ) & 
        \bfM_{\uh }(\uh )\bfM_{\ph}(\ph+{\gh}_Z+{\hh}_Z)
    \end{bmatrix},
\end{align*}
where we have used $\ph = {\gh}_X + {\hh}_X + g_{x_0}$. Finally, the CNOT operation at Line~\ref{line:MZWpOdd} gives:
\begin{align*}
    \xrightarrow{\mathrm{CNOT}(D_u,\, Z[\bfM_{\ph}(h_{z_0})])}
    \begin{bmatrix}
        I & \bfM_{\ph}(\ph) & \bfM_{\uh }(\uh ) & 0\\
        0 & \bfM_{\ph}(\ph) & \bfM_{\uh }(\uh ) & 0
    \end{bmatrix},
\end{align*}
where we have used $\ph = {\gh}_Z + {\hh}_Z + h_{z_0}$ and ${\gh}_X+g_{x_0}={\hh}_Z+h_{z_0}$.

\medskip
\noindent\textbf{Case 2: $w_{\ph}$ even.}
In this case, $\ph = {\gh}_X + {\hh}_X = {\gh}_Z + {\hh}_Z$ 
and $\uh  = {\fh}$ by construction. Applying $T$ directly to the 
initial tableau yields:
\begin{align*}
    \xrightarrow{T}
    \begin{bmatrix}
        I & \bfM_{\ph}(\ph) & \bfM_{\uh }(\uh ) & 0\\
        0 & \bfM_{\ph}(\ph) & \bfM_{\uh }(\uh ) & 0
    \end{bmatrix}.
\end{align*}
In all the cases, the desired tableau transformation is achieved.

Now we will show that the action of our syndrome extraction circuit on logical qubits is trivial. Let us consider a logical operator containing $X$ Paulis with the form $(0,d_p,d_u,0)$. Let us track this logical operator through our chain of transformations. For the case where $w_{\ph}$ is even, applying the general transformation matrix $T$ yields a new vector:
\begin{align*}
    (0, d_p, d_u, 0) T &= (0, d_p, d_u, d_p\bfM_{\uh }({\fh}) + d_u\bfM_{\ph}({\gh}_Z + {\hh}_Z)).
\end{align*}
Note that the commutativity between $X$-type logicals and $Z$-type stabilizers implies:
\begin{align*}
    d_p\bfM_{\uh }(\uh ) + d_u\bfM_{\ph}(\ph) = 0.
\end{align*}
Combining this with the fact that $\uh  = {\fh}$ and $\ph = {\gh}_Z + {\hh}_Z$, we find:
\begin{align*}
    (0, d_p, d_u, 0) T &= (0, d_p, d_u, 0).
\end{align*}

For the case where $w_{\ph}$ is odd and $(S,U)=(Z,X)$, the initial $\text{CNOT}(D_u, Z[\bfM_{\ph}(g_{z_0})])$ transforms the initial vector to $(0, d_p, d_u, d_u\bfM_{\ph}(g_{z_0}))$. Applying the transformation matrix $T$ to this intermediate state yields:
\begin{align*}
    (0, d_p, d_u, d_u\bfM_{\ph}(g_{z_0})) T &= (0, d_p, d_u, t),
\end{align*}
where
\begin{align*}
    t &= d_u\bfM_{\ph}(g_{z_0}) + d_p\bfM_{\uh }({\fh}) + d_u\bfM_{\ph}({\gh}_Z + {\hh}_Z) \\
    &= d_p\bfM_{\uh }(\uh ) + d_u\bfM_{\ph}({\gh}_Z + {\hh}_Z + g_{z_0}) \\
    &= d_p\bfM_{\uh }(\uh ) + d_u\bfM_{\ph}(\ph) = 0.
\end{align*}
The final operation in the odd-weight schedule is $\text{CNOT}(X, D_p[\bfM_{\ph}(h_{x_0})])$. Because this operation relies on the $X$ ancilla as the control, and the $X$ register has remained identically $0$, this operation has a trivial action, and the state remains $(0, d_p, d_u, 0)$. One can conclude the same for the case $w_{\ph}$ is odd and $(S,U)=(X,Z)$ by following the same sequence of steps.

In all the cases, we have shown that the SC circuit maps the vector $(0, d_p, d_u, 0)$ to itself. Hence, the circuit acts trivially on $X$-type logical operators.

The correctness of the $Z$-type syndrome extraction and its trivial action on $Z$-type logical operators follows by symmetry. To verify this, one must simply update the propagation rule and the initial tableau to track $Z$-type Pauli operators. For $Z$-type stabilizers, the collective action of parallel CNOT gates defined by $\bfM$ on $Z$-type blocks transforms as 
\begin{align*}
    (\bfC, \bfT) \xrightarrow{\bfM} (\bfC + \bfT\bfM^T, \bfT)
\end{align*}
due to CNOT actions on $I\otimes Z \to Z\otimes Z$ and $Z\otimes I \to Z\otimes I$. Following a similar sequence of steps, one can show the required tableau transformation:
\begin{align*}
    \begin{bmatrix} 0 & 0 & 0 & I\\ 0 & \bfM_{\uh }(\uh )^T & \bfM_{\ph}(\ph)^T & 0 \end{bmatrix} \longrightarrow \begin{bmatrix} 0 & \bfM_{\uh }(\uh )^T & \bfM_{\ph}(\ph)^T & I\\ 0 & \bfM_{\uh }(\uh )^T & \bfM_{\ph}(\ph)^T & 0 \end{bmatrix},
\end{align*}
completing the proof of Proposition~\ref{prop:SyndromeCircuit}

\section{Counting Argument Regarding the Sequence Configurations}\label{appendix:Counting_Proof}
The elements of $\operatorname{supp}(\uh )$ are arranged into two independent sequences, $\vec{f}_X$ and $\vec{f}_Z$. Since there are no constraints on $\uh $, there are $w_{\uh }!$ valid orderings for each sequence, yielding  $(w_{\uh }!)^2$ configurations. The partitioning of $\operatorname{supp}(\ph)$ depend on the parity of $w_{\ph}$:

\textbf{Case 1: $w_{\ph}$ is even.}
The sequences $\vec{g}_X$ and $\vec{h}_X$ partition $\operatorname{supp}(\ph)$ into two equal-length ordered sequences. The total number of such ordered pairs equals $w_{\ph}!$. The sequences $\vec{g}_Z$ and $\vec{h}_Z$ must partition the same support such that $\vec{h}_Z$ contains the same elements as $\vec{g}_X$. Because these subsets are uniquely fixed by the $X$-partition, we only choose their internal orderings. There are $(w_{\ph}/2)!$ ways to order each, yielding $((w_{\ph}/2)!)^2$ configurations. Noting the identity $\binom{w_{\ph}}{w_{\ph}/2} = w_{\ph}! / ((w_{\ph}/2)!)^2$, the number of $Z$-sequence choices simplifies to $w_{\ph}! / \binom{w_{\ph}}{w_{\ph}/2}$. Multiplying the $X$ and $Z$ choices gives the total sequence pairs for $\ph$:
\begin{align*}
   \frac{(w_{\ph}!)^2}{\binom{w_{\ph}}{w_{\ph}/2}} 
\end{align*}
\textbf{Case 2: $w_{\ph}$ is odd.}
The algorithm first selects a basis $S \in \{X,Z\}$. The sequences $\vec{g}_S$ and $\vec{h}_S$ partition $\operatorname{supp}(\ph)$ such that $|\vec{g}_S| = \lfloor w_{\ph}/2 \rfloor + 1$ and $|\vec{h}_S| = \lfloor w_{\ph}/2 \rfloor$. The number of such ordered pairs is $w_{\ph}!$. The sequences $\vec{g}_U$ and $\vec{h}_U$ are constrained such that $\vec{h}_U$ contains the exact elements of $\vec{g}_S$. The internal orderings for these fixed subsets yield $(\lfloor w_{\ph}/2 \rfloor + 1)! \lfloor w_{\ph}/2 \rfloor!$ choices, which equals to $w_{\ph}! / \binom{w_{\ph}}{\lfloor w_{\ph}/2 \rfloor}$. Multiplying these choices gives the total sequence pairs for $\ph$:
\begin{align*}
   2 \frac{(w_{\ph}!)^2}{\binom{w_{\ph}}{\lfloor w_{\ph}/2 \rfloor}}, 
\end{align*}  
where the factor $2$ arises from the different cases of the variable $S$. Therefore, combining the independent choices for $\uh $ and $\ph$ for a given $\ph$ yields:
\begin{align*}
(w_{\uh }!)^2 \times \nu(w_{\ph}) \frac{(w_{\ph}!)^2}{\binom{w_{\ph}}{\lfloor w_{\ph}/2 \rfloor}} = (w_{\uh }! w_{\ph}!)^2 \frac{\nu(w_{\ph})}{\binom{w_{\ph}}{\lfloor w_{\ph}/2 \rfloor}}
\end{align*}
Finally, summing this expression over two cases where, $\ph = \ah$ and $\ph = \bh$,  yields the total number of  configurations:
\begin{align*}
(w_{\ah}! w_{\bh}!)^2 \left(\frac{\nu(w_{\ah})}{\binom{w_{\ah}}{\lfloor w_{\ah}/2 \rfloor}} + \frac{\nu(w_{\bh})}{\binom{w_{\bh}}{\lfloor w_{\bh}/2 \rfloor}}\right)
\end{align*}

\section{Fit Parameters}\label{appendix:FitParams}
In Table \ref{tab:fit_parameters}, we report the fit parameters used in the heuristic formula $\Tilde{p}_L(p)= p^{d_{\mathrm{circ}}/2}e^{\alpha+\beta p + \gamma p^2}$ for each code listed in Table \ref{tab:group_codes}.
\begin{table}[htbp]
    \centering
    \caption{Fit Parameters of the codes listed in Table \ref{tab:group_codes}}
    \label{tab:fit_parameters}
    \begin{tabular}{c|c|c|c|c} 
        \toprule
        Codes & $d_{\mathrm{circ}}$ & $\alpha$ & $\beta$ & $\gamma$ \\ 
        \midrule
        $[[48,8,6]]$ & $\leq 5$ & 7.69 & 1363 & -125961 \\
        $[[96,8,10]]$ & $\leq 9$ & 15.71 & 2259 & -241331 \\
        $[[224,12,16]]$ & $\leq 14$ & 21.18 & 4129 & -327354 \\
        $[[84,16,8]]$ & $\leq 7$ & 12.53 & 2252 & -250402 \\ 
        $[[112,16,10]]$ & $\leq 8$ & 13.38 & 3263 & -377805 \\ 
        $[[128,16,12]]$ & $\leq 10$ & 20.37 & 2308 & -241828 \\ 
        $[[168,16,15]]$ & $\leq 13$ & 30.34 & 932 & 24829 \\ 
        \bottomrule
    \end{tabular}
\end{table}

\section{Additional Code Examples}
We present additional examples of codes in Table \ref{tab:group_codes_additional} below.

\begin{table*}[htbp]
\centering
\caption{Additional examples of codes. Distances reported with a $\leq$ symbol represent upper bounds obtained using \texttt{QDistRnd} with $10^6$ iterations. All other distances are exact.}
\label{tab:group_codes_additional}
\begin{tabular}{@{} ccc cc cc ccc @{}}
\toprule
\multicolumn{3}{c}{\textbf{Code Params}} & \multicolumn{2}{c}{\textbf{Group Descriptions}} & \multicolumn{2}{c}{\textbf{Group Algebra}} & \multicolumn{3}{c}{\textbf{GAP Identifiers}} \\
\cmidrule(lr){1-3} \cmidrule(lr){4-5} \cmidrule(lr){6-7} \cmidrule(lr){8-10}

\makebox[1cm][c]{$n$} & \makebox[1cm][c]{$k$} & \makebox[1cm][c]{$d$} &
$G$ Structure & $H$ Structure & $\ah$ & $\bh$ &
\makebox[1cm][c]{$\ell$} & \makebox[1cm][c]{$m$} & \makebox[1cm][c]{$s$} \\ \midrule

\multicolumn{10}{c}{\textbf{Weight-6 codes}} \\[4pt]

36 & 4 & 6 & $C_9 \times D_{18}$ & $C_9$ & $[1, 15, 36]$ & $[1, 4, 8]$ & 162 & 3 & 24 \\
48 & 4 & 8 & $C_3 \times (C_9 \rtimes C_8)$ & $C_9$ & $[1, 26, 69]$ & $[1, 8, 9]$ & 216 & 12 & 14 \\
54 & 8 & 6 & $C_{36} \times S_3$ & $C_4 \times C_2$ & $[1, 13, 27]$ & $[1, 4, 5]$ & 216 & 47 & 23 \\
56 & 6 & 8 & $C_7 \times D_8$ & $C_2$ & $[1, 5, 56]$ & $[1, 2, 9]$ & 56 & 9 & 1 \\
60 & 16 & 4 & $C_{15} \times D_{10}$ & $C_5$ & $[1, 11, 31]$ & $[1, 10, 12]$ & 150 & 8 & 6 \\
72 & 4 & 10 & $C_3 \times ((C_6 \times C_2) \rtimes C_2)$ & $C_2$ & $[1, 12, 50]$ & $[1, 4, 17]$ & 72 & 30 & 1 \\
72 & 8 & 8 & $C_{36} \times S_3$ & $C_6$ & $[1, 6, 31]$ & $[1, 7, 10]$ & 216 & 47 & 17 \\
84 & 6 & 10 & $C_{21} \times (C_3 \rtimes C_4)$ & $C_6$ & $[1, 24, 34]$ & $[1, 10, 20]$ & 252 & 21 & 6 \\
90 & 8 & 10 & $C_{15} \times ((C_6 \times C_2) \rtimes C_2)$ & $D_8$ & $[1, 39, 41]$ & $[1, 4, 14]$ & 360 & 99 & 35 \\
96 & 4 & 12 & $(C_3 \times C_3) \rtimes ((C_4 \times C_4) \rtimes C_2)$ & $S_3$ & $[1, 19, 52]$ & $[1, 14, 24]$ & 288 & 489 & 93 \\
108 & 8 & 10 & $C_{36} \times S_3$ & $C_4$ & $[1, 13, 67]$ & $[1, 8, 18]$ & 216 & 47 & 12 \\
112 & 6 & 12 & $C_7 \times D_8 \times S_3$ & $S_3$ & $[1, 76, 107]$ & $[1, 13, 14]$ & 336 & 188 & 59 \\
112 & 12 & 8 & $C_7 \times D_8 \times S_3$ & $S_3$ & $[1, 44, 106]$ & $[1, 12, 27]$ & 336 & 188 & 59 \\
120 & 8 & 12 & $C_{15} \times ((C_6 \times C_2) \rtimes C_2)$ & $C_6$ & $[1, 162, 211]$ & $[1, 7, 11]$ & 360 & 99 & 18 \\
126 & 10 & 10 & $S_3 \times (C_7 \rtimes C_3)$ & $C_2$ & $[1, 10, 81]$ & $[1, 5, 19]$ & 126 & 8 & 1 \\
144 & 4 & 16 & $C_3 \times (C_9 \rtimes C_8)$ & $C_3$ & $[1, 26, 179]$ & $[1, 22, 33]$ & 216 & 12 & 1 \\
150 & 16 & 8 & $C_{15} \times D_{10}$ & $C_2$ & $[1, 4, 55]$ & $[1, 7, 15]$ & 150 & 8 & 1 \\
168 & 16 & 10 & $C_7 \times ((C_6 \times C_2) \rtimes C_2)$ & $C_2$ & $[1, 92, 118]$ & $[1, 14, 30]$ & 168 & 33 & 1 \\
180 & 8 & 16 & $C_{15} \times ((C_6 \times C_2) \rtimes C_2)$ & $C_4$ & $[1, 84, 128]$ & $[1, 5, 22]$ & 360 & 99 & 14 \\
186 & 10 & 14 & $C_{31} \times S_3$ & $C_2$ & $[1, 74, 183]$ & $[1, 21, 29]$ & 186 & 3 & 1 \\
192 & 8 & 16 & $C_3 \times ((C_{16} \rtimes C_2) \rtimes C_4)$ & $C_4$ & $[1, 226, 283]$ & $[1, 9, 21]$ & 384 & 512 & 8 \\
248 & 10 & $\leq 18$ & $C_{31} \times D_8$ & $C_2$ & $[1, 122, 235]$ & $[1, 15, 30]$ & 248 & 9 & 1 \\

\midrule
\multicolumn{10}{c}{\textbf{Weight-8 codes}} \\[4pt]

48 & 10 & 6 & $C_3 \times ((C_6 \times C_2) \rtimes C_2)$ & $C_3$ & $[1, 16, 57, 58]$ & $[1, 2, 3, 9]$ & 72 & 30 & 9 \\
72 & 16 & 6 & $C_3 \times ((C_6 \times C_2) \rtimes C_2)$ & $C_2$ & $[1, 3, 43, 47]$ & $[1, 6, 8, 12]$ & 72 & 30 & 1 \\
80 & 8 & 10 & $(C_5 \rtimes C_8) \rtimes C_2$ & $C_2$ & $[1, 3, 43, 49]$ & $[1, 9, 16, 18]$ & 80 & 10 & 1 \\
80 & 10 & 8 & $(C_5 \rtimes C_8) \rtimes C_2$ & $C_2$ & $[1, 10, 52, 58]$ & $[1, 3, 9, 14]$ & 80 & 10 & 1 \\
84 & 10 & 9 & $C_{21} \times (C_3 \rtimes C_4)$ & $C_6$ & $[1, 82, 101, 120]$ & $[1, 8, 13, 16]$ & 252 & 21 & 6 \\
90 & 18 & 7 & $C_5 \times ((C_3 \times C_3) \rtimes C_3)$ & $C_3$ & $[1, 78, 111, 112]$ & $[1, 4, 6, 7]$ & 135 & 3 & 1 \\
96 & 10 & 12 & $(C_3 \times C_3) \rtimes ((C_4 \times C_4) \rtimes C_2)$ & $S_3$ & $[1, 18, 28, 84]$ & $[1, 20, 22, 23]$ & 288 & 489 & 93 \\
96 & 12 & 10 & $(C_3 \times C_3) \rtimes ((C_4 \times C_4) \rtimes C_2)$ & $C_6$ & $[1, 77, 85, 136]$ & $[1, 13, 18, 21]$ & 288 & 489 & 103 \\
96 & 18 & 8 & $(C_3 \times C_3) \rtimes ((C_4 \times C_4) \rtimes C_2)$ & $S_3$ & $[1, 19, 31, 82]$ & $[1, 2, 18, 24]$ & 288 & 489 & 93 \\
108 & 12 & 10 & $((C_9 \times C_3) \rtimes C_3) \rtimes C_2$ & $C_3$ & $[1, 29, 92, 97]$ & $[1, 2, 12, 17]$ & 162 & 4 & 13 \\
112 & 12 & 12 & $C_7 \times ((C_4 \times C_4) \rtimes C_2)$ & $C_4$ & $[1, 46, 82, 107]$ & $[1, 17, 20, 25]$ & 224 & 53 & 8 \\
120 & 14 & 12 & $C_{15} \times D_8$ & $C_2$ & $[1, 13, 47, 100]$ & $[1, 13, 16, 19]$ & 120 & 32 & 3 \\
120 & 16 & 11 & $C_{15} \times D_8$ & $C_2$ & $[1, 9, 93, 101]$ & $[1, 2, 15, 30]$ & 120 & 32 & 3 \\
120 & 24 & 7 & $C_5 \times ((C_3 \times C_3) \rtimes C_4)$ & $C_3$ & $[1, 32, 69, 73]$ & $[1, 9, 20, 23]$ & 180 & 23 & 12 \\
128 & 10 & 14 & $(C_8 \rtimes C_2) \rtimes C_8$ & $C_2$ & $[1, 33, 92, 120]$ & $[1, 18, 28, 30]$ & 128 & 10 & 1 \\
128 & 18 & 10 & $(C_8 \rtimes C_2) \rtimes C_8$ & $C_2$ & $[1, 40, 75, 107]$ & $[1, 8, 12, 19]$ & 128 & 10 & 1 \\
128 & 30 & 8 & $(C_2 \times (C_4 \rtimes C_4)) \rtimes C_4$ & $C_2$ & $[1, 21, 78, 82]$ & $[1, 12, 22, 27]$ & 128 & 26 & 7 \\
136 & 8 & 15 & $C_{17} \times D_8$ & $C_2$ & $[1, 11, 55, 61]$ & $[1, 12, 15, 29]$ & 136 & 10 & 1 \\
144 & 10 & 15 & $C_3 \times (C_9 \rtimes C_8)$ & $C_3$ & $[1, 17, 173, 200]$ & $[1, 8, 33, 35]$ & 216 & 12 & 1 \\
144 & 18 & 12 & $(C_3 \times C_3) \rtimes ((C_4 \times C_2 \times C_2) \rtimes C_2)$ & $C_4$ & $[1, 87, 88, 117]$ & $[1, 7, 29, 35]$ & 288 & 498 & 141 \\
144 & 24 & 8 & $(C_3 \times C_3) \rtimes ((C_4 \times C_4) \rtimes C_2)$ & $C_4$ & $[1, 112, 129, 134]$ & $[1, 12, 17, 20]$ & 288 & 489 & 15 \\
160 & 8 & $\leq 18$ & $C_5 \rtimes ((C_8 \times C_4) \rtimes C_2)$ & $C_4$ & $[1, 28, 108, 113]$ & $[1, 13, 17, 38]$ & 320 & 17 & 61 \\
160 & 12 & 16 & $C_5 \rtimes ((C_4 \rtimes C_4) \rtimes C_4)$ & $C_2 \times C_2$ & $[1, 65, 105, 128]$ & $[1, 14, 28, 31]$ & 320 & 10 & 59 \\
168 & 10 & $\leq17$ & $C_7 \times ((C_6 \times C_2) \rtimes C_2)$ & $C_2$ & $[1, 60, 71, 108]$ & $[1, 7, 19, 36]$ & 168 & 33 & 1 \\
168 & 24 & 10 & $D_8 \times (C_7 \rtimes C_3)$ & $C_2$ & $[1, 10, 48, 49]$ & $[1, 12, 15, 25]$ & 168 & 20 & 3 \\
186 & 22 & 12 & $C_{31} \times S_3$ & $C_2$ & $[1, 23, 90, 180]$ & $[1, 2, 9, 26]$ & 186 & 3 & 1 \\
186 & 36 & 7 & $C_{31} \times S_3$ & $C_2$ & $[1, 4, 19, 67]$ & $[1, 6, 13, 19]$ & 186 & 3 & 1 \\
200 & 28 & 10 & $C_5 \times C_5 \times D_8$ & $C_2$ & $[1, 25, 157, 172]$ & $[1, 14, 16, 33]$ & 200 & 38 & 3 \\
208 & 10 & $\leq 20$ & $C_{13} \times ((C_4 \times C_2) \rtimes C_2)$ & $C_2$ & $[1, 8, 172, 193]$ & $[1, 24, 30, 50]$ & 208 & 21 & 1 \\
216 & 14 & $\leq 18$ & $C_9 \times ((C_6 \times C_2) \rtimes C_2)$ & $C_2$ & $[1, 104, 196, 215]$ & $[1, 6, 19, 35]$ & 216 & 58 & 7 \\
240 & 16 & $\leq 20$ & $C_3 \times ((C_5 \rtimes C_8) \rtimes C_2)$ & $C_2$ & $[1, 37, 136, 228]$ & $[1, 11, 35, 52]$ & 240 & 39 & 1 \\

\bottomrule
\end{tabular}
\end{table*}
\end{document}